\documentclass{ieeetmlcn}
\usepackage{cite}
\usepackage{amsmath,amssymb,amsfonts}
\usepackage{algorithmic}
\usepackage{graphicx,color}
\usepackage{textcomp}
\usepackage{xcolor}
\usepackage{hyperref}
\hypersetup{hidelinks=true}
\usepackage{algorithm,algorithmic}
\def\BibTeX{{\rm B\kern-.05em{\sc i\kern-.025em b}\kern-.08em
		T\kern-.1667em\lower.7ex\hbox{E}\kern-.125emX}}
\AtBeginDocument{\definecolor{tmlcncolor}{cmyk}{0.93,0.59,0.15,0.02}\definecolor{NavyBlue}{RGB}{0,86,125}}

\def\OJlogo{\vspace{-4pt}\includegraphics[height=18pt]{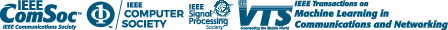}}
\def\seclogo{\vspace{10pt}\includegraphics[height=18pt]{new_logo}}

\usepackage{cite}
\usepackage{amsmath,amssymb,amsfonts,amsthm}
\usepackage{dsfont}
\usepackage{bbm}
\usepackage{graphicx}
\usepackage{textcomp}
\usepackage{xcolor}
\let\labelindent\relax
\usepackage{enumitem}
\usepackage{bm}
\usepackage{hyperref}
\usepackage{booktabs}
\usepackage{algorithm}
\usepackage{algorithmic}
\usepackage{tikz}
\usepackage{multirow}
\usetikzlibrary{positioning, arrows.meta, calc}
\usetikzlibrary{decorations.pathreplacing}
\usetikzlibrary{shapes.geometric}

\newtheorem{definition}{Definition}
\newtheorem{theorem}{Theorem}
\newtheorem{lemma}{Lemma}

\newtheorem{assumption}{Assumption}
\newtheorem*{theorem*}{Theorem}
\usepackage{thmtools}
\usepackage{thm-restate}
\newtheorem{corollary}{Corollary}
\newtheorem{remark}{Remark}

\usepackage{microtype}
\usepackage{xspace}
\newcommand{\name}{\textls[80]{LightTune}\xspace}
\newcommand{\CQI}{\textls[80]{CQI-Tune}\xspace}
\newcommand{\RICQI}{\textls[80]{RI-CQI-Tune}\xspace}
\newcommand{\BLER}{\textls[80]{BLER-Predict}\xspace}
\usepackage{subcaption}
\usepackage{array}

\newcommand{\blfootnote}[1]{%
	\begingroup
	\renewcommand{\thefootnote}{}%
	\footnotetext{#1}%
	\endgroup
}

\def\BibTeX{{\rm B\kern-.05em{\sc i\kern-.025em b}\kern-.08em
		T\kern-.1667em\lower.7ex\hbox{E}\kern-.125emX}}

\begin{document}
	\IEEEoverridecommandlockouts
\receiveddate{May, 2026}
\reviseddate{XX Month, XXXX}
\accepteddate{XX Month, XXXX}
\publisheddate{XX Month, XXXX}
\currentdate{May, 2026}
\doiinfo{TMLCN.2026.XXXXXXX}

\markboth{}{Ali and Penna}
	
	\title{LightTune: Lightweight Forward-Only Online Fine-Tuning with Applications to Link Adaptation}
	\author{\IEEEauthorblockN{Ramy E. Ali,~\IEEEmembership{Member,~IEEE,} and Federico Penna,~\IEEEmembership{Member,~IEEE}}%
	}

	\begin{abstract}
		Deploying machine learning (ML) algorithms on mobile phones is bottlenecked by performance degradation under dynamic, real-world conditions that differ from the offline training conditions. While continual learning and adaptation are essential to mitigate this distributional shift, conventional online learning methods are often computationally prohibitive for resource-constrained devices. In this paper, we propose \name{}, a lightweight, backpropagation-free online fine-tuning framework with provable convergence guarantees. \name{} opportunistically refines ML models using live test-time data only when performance falls below a predefined threshold, ensuring minimal computational overhead and efficient responsiveness. As a practical demonstration, we integrate \name{} into a block error rate (BLER) prediction algorithm for 6G mobiles. This integration enables the ML BLER prediction model to dynamically adapt to previously unseen channel conditions in real time. Simulation results show a substantial reduction in the average BLER prediction error by up to $43.5\%$ with online fine-tuning. Furthermore, we leverage this BLER prediction algorithm for link adaptation and demonstrate average throughput improvements by up to $15.5\%$ compared to a conventional table-based outer loop link adaptation (OLLA) algorithm.
		
	\end{abstract}

	\maketitle
	\IEEEpeerreviewmaketitle
		\blfootnote{This article was presented in part at the proceedings of
	the 2026 International Conference on Communications (ICC)
	\cite{Ali2026lighttune}. The authors are with Samsung Device
	Solutions Research America, Samsung Semiconductor, Inc., San Diego,
	CA 92121 USA (e-mails:
	\href{mailto:ramy.ali@samsung.com}{ramy.ali@samsung.com},
	\href{mailto:f.penna@samsung.com}{f.penna@samsung.com}).}
	\section{Introduction}
	\IEEEPARstart{M}{achine} learning (ML)-based wireless algorithms are emerging as pivotal enablers for 6G applications, including channel state information (CSI) prediction, compression, and beam management \cite{3gppTR38843}. Increasingly, such algorithms are deployed directly on user equipment (UE) and edge devices for real-time inference, including as components of digital twin networks (DTNs) \cite{zhang2025digital,lin20236g}. However, a fundamental bottleneck is the inevitable training-test distribution mismatch: ML models trained offline on synthetic data fail to capture the full complexity of real-world environments, and the highly dynamic nature of wireless channels causes significant post-deployment performance degradation \cite{kaswan2024statistical,luo2025digital}.
	
	To mitigate this training-test mismatch, ML models must continuously adapt to dynamic deployment environments through online learning \cite{xulearning,Samsung2025AI6GR}. However, conventional online learning is often infeasible for resource-constrained devices, such as UE and edge devices, which typically lack the computational and memory capacities required to support real-time gradient computations and backpropagation.
	
	In this paper, we propose \name{}, a lightweight online fine-tuning framework that enables the UE to incrementally refine deployed ML models using inference-time observations. \name{} is specifically designed for prediction tasks where ground-truth metrics, such as link-level throughput or block error rate (BLER), become available after a brief delay. Such metrics are inherently available in wireless communications and networked systems. By leveraging this delayed ground-truth data, \name{} performs targeted updates to continuously improve model accuracy in a direct, supervised manner. Unlike reinforcement learning (RL)-based approaches, which rely on reward signals and unsupervised exploration \cite{saxena2021reinforcement,an2024dragon}, \name{} utilizes direct ground-truth labels to ensure convergence. Crucially, while standard adaptation techniques rely on computationally expensive backpropagation \cite{xu2024learning}, \name{} is entirely backpropagation-free, making it ideally suited for a resource-limited UE, particularly wireless modems. \\
	\textbf{Contributions}. Our main contributions are as follows.\\ 1) We develop a fine-tuning algorithm termed \name{} that is opportunistically triggered when the performance of the ML model, initially trained offline, is not as desired. The performance of the ML model is \emph{monitored} in terms of the prediction error, and once it reaches a predefined threshold $\delta$, the fine-tuning procedure is triggered. \name{} offers backpropagation-free online fine-tuning by leveraging the forward-forward (FF) algorithm~\cite{hinton2022forward}, enhanced by a newly proposed lightweight smooth loss function, which enables low-complexity online gradient computation, and buffer-less fine-tuning through a proposed threshold-based update policy. This policy decides which samples are used on a sample-by-sample basis, without storing them  first \cite{lin1992self,schaul2015prioritized,rolnick2019experience}. To the best of our knowledge, \name{} is the first application of the FF algorithm in cellular wireless communications. \\
	2) We provide finite-time and asymptotic convergence guarantees for \name{}  under stochastic gradient descent (SGD) showing that, under training-test distribution mismatch, the average frequency of prediction errors reaching or exceeding any fixed threshold $\delta$ converges to $0$ as the number of fine-tuning steps increases.\\
		3) To demonstrate the practical utility of \name{}, we evaluate its performance within a short-term BLER prediction framework for wireless cellular systems. Integrating this BLER predictor into the link adaptation process yields throughput gains reaching up to $15.5\%$ relative to a conventional table-based outer-loop link adaptation (OLLA) baseline. 
		
			\textbf{Organization}. We discuss background concepts and prior work in Section \ref{sec:background}. 
	\name{} is  presented in Section \ref{sec:proposed} and the convergence results are provided in Section \ref{sec:convergence}. \name{} is then leveraged for BLER prediction and link adaptation in Section \ref{sec:app}. We present our  simulation results in Section \ref{sec:results} and discuss concluding remarks and future work in  Section \ref{sec:conclusion}.
	\section{Background and Related Work}
	\label{sec:background}
	We begin by providing an overview of the FF algorithm, which is commonly used for computer vision problems \cite{hinton2022forward}.
	\subsection{The Forward-Forward (FF) Algorithm}
	In contrast to backpropagation (BP), the FF algorithm trains each layer locally and sequentially, avoiding the need to propagate derivatives backward through the network. Based on two forward passes \cite{hinton2022forward}, its layer-local update rules eliminate the need for a global computational graph, making it a natural fit for resource-constrained UE. The two forward passes are described as follows.
	\begin{enumerate}[leftmargin=*, itemsep=0pt, parsep=0pt]
		\item \textbf{Forward Positive Pass}.
		A positive data sample is defined as the tuple $  (\bm x, y_{\text +})$, where $\bm x \in \mathbb{R}^d$ is the input feature vector and $y_{\text +} \in \mathcal Y = \{y_1, y_2, \dots, y_C\}$ is the corresponding true label. The positive pass operates on the positive data samples and optimizes the model parameters to increase a goodness value in every layer above a predefined threshold $T$, where $T>0$.
		\item \textbf{Forward Negative Pass}. 
		A negative data sample is defined as $ (\bm x, y_{\text{-}})$, where $y_{\text -} \in \mathcal{Y} \setminus{ \{y_{\text{+}}\}}$ is an incorrect label drawn from a distribution over incorrect labels (e.g., uniform). The negative pass operates on negative samples and adjusts the model  to decrease a goodness value below the threshold $T$. 
	\end{enumerate}
	
	\begin{itemize}[leftmargin=*, itemsep=4pt, parsep=0pt]
		\item \textbf{Label Encoding.} We adopt a concatenation-based encoding, where the label $y$ is appended to the feature vector $\bm{x} \in \mathbb{R}^{d}$ to form an augmented input vector $\bm{z} = [\bm{x}^\top, y]^\top \in \mathbb{R}^{d+1}$.
		
		\item \textbf{End-to-end Mapping.} We denote the  model parameters as $\bm{\theta} = [\bm{\theta}_1^\top, \bm{\theta}_2^\top, \dots, \bm{\theta}_L^\top]^\top$, where $\bm{\theta}_l \in \mathbb{R}^{d_l}$ represents the parameter vector of layer $l \in \{1, 2, \dots, L\}$. Each layer implements a transformation $f_l(\cdot; \bm{\theta}_l)$ that maps an input representation to a subsequent hidden state. For a given feature vector $\bm{x} \in \mathbb{R}^d$ and a candidate scalar label $y \in \mathcal{Y}$, we construct the augmented input $\bm{z} = [\bm{x}^\top, y]^\top \in \mathbb{R}^{d+1}$, which serves as the initial activation $\bm{h}_0$. The model's forward pass is defined by the recursive composition of these layer-wise functions, $\bm{h}_l = f_l(\bm{h}_{l-1}; \bm{\theta}_l)$, such that the final activation vector $\bm{h}_L$ is expressed as
		\begin{equation}
			\bm{h}_L := F_{\bm{\theta}}^{(L)}(\bm{z}) = (f_L \circ f_{L-1} \circ \cdots \circ f_1)(\bm{z}),
		\end{equation}
		where $F_{\bm{\theta}}^{(l)}: \mathbb{R}^{d+1} \to \mathbb{R}^{M_l}$ denotes the cumulative mapping from the initial augmented input space to the activation space of the $l$-th layer, and $M_l$ represents the output dimensionality of layer $l$. By convention, the base case is defined as $\bm{h}_0 = F_{\bm{\theta}}^{(0)}(\bm{z}) = \bm{z} \in \mathbb{R}^{M_0}$ where $M_0 = d+1$.
		
		\item \textbf{Goodness Function.} The compatibility between the features $\bm{x}$ and the candidate label $y$ is quantified by a goodness function $G_{\bm \theta}(\bm{x}, y)$, defined as 
		\begin{equation}
			\label{eqn:goodness}
			G_{\bm \theta}(\bm{x}, y) = \|\bm{h}_L\|_2^2 = \| F_{\bm{\theta}}^{(L)}([\bm{x}^\top, y]^\top) \|_2^2,
		\end{equation}
		where the learning objective is to maximize $G_{\bm \theta}(\bm{x}, y)$ for positive samples (ground-truth pairings) and minimize it for negative samples (incorrect or corrupted pairings). 
		
		\item \textbf{Softplus Loss Function.} Let $\bm{h}_{\text{+},l}^{(i)}$ and $\bm{h}_{\text{-},l}^{(i)} \in \mathbb{R}^{M_l}$ denote the activation vectors of layer $l$ for the $i$-th positive and negative samples, respectively. The per-neuron positive and negative goodness values are then defined as
		\begin{align}
			g_{\text{+},l}^{(i)}[j] = (h_{\text{+},l}^{(i)}[j])^2, \  g_{\text{-},l}^{(i)}[j] = (h_{\text{-},l}^{(i)}[j])^2,
		\end{align}
		for each neuron $j \in \{1, 2, \dots, M_l \}$. The positive and negative goodness vectors for layer $l$ are then given as $\bm g_{\text{+},l}^{(i)} = \left[g_{\text{+},l}^{(i)}[1], g_{\text{+},l}^{(i)}[2], \dots, g_{\text{+},l}^{(i)}[M_l] \right]^\top$ and $\bm g_{\text{-},l}^{(i)} = \left[g_{\text{-},l}^{(i)}[1], g_{\text{-},l}^{(i)}[2], \dots, g_{\text{-},l}^{(i)}[M_l] \right]^\top$. It is worth noting that $G_{\bm \theta}(\bm{x}, y)$ defined in (\ref{eqn:goodness}) corresponds to the sum of these per-neuron goodness values at the terminal layer $L$, such that $G_{\bm \theta}(\bm{x}, y) = \sum_{j=1}^{M_L} g_{L}[j] = \sum_{j=1}^{M_L}  (h_{L}[j])^2$.
		
		For a set of positive samples $\mathcal{S}_{\text{+}}$ and negative samples $\mathcal{S}_{\text{-}}$, the Softplus loss for layer $l$ is expressed as
		\begin{align}
			\mathcal{L}_{\text{Softplus},l} &= \frac{1}{|\mathcal{S}_{\text{+}}| M_l} \sum_{i \in \mathcal{S}_{\text{+}}} \sum_{j=1}^{M_l} \ln \left(1 + e^{-(g_{\text{+},l}^{(i)}[j] - T)}\right) \notag \\
			&\quad + \frac{1}{|\mathcal{S}_{\text{-}}| M_l} \sum_{i \in \mathcal{S}_{\text{-}}} \sum_{j=1}^{M_l} \ln \left(1 + e^{g_{\text{-},l}^{(i)}[j] - T}\right),
		\end{align}
		where $T$ is a fixed threshold. This loss function encourages the goodness values of individual neurons to stay above $T$ for positive samples and below $T$ for negative samples.
		
		\item \textbf{Inference.} 
		The predicted label $\hat{y}$ for a given feature vector $\bm{x}$ is obtained  as  
		\begin{align}
			\hat{y} = \arg\max_{y \in \mathcal{Y}} G_{\bm \theta}(\bm{x}, y) 
			= \arg\max_{y \in \mathcal{Y}} \left\| F_{\bm \theta}^{(L)} \left( [\bm{x}^\top, y]^\top \right) \right\|_2^2.
		\end{align}
	\end{itemize}
	
	\begin{remark}[\textbf{Output layer in the FF algorithm}]
		Unlike conventional classifiers, where the output layer's dimension is equal to the number of classes $C$, the output layer's dimension may not be equal to $C$ in the FF algorithm as it predicts the label based on goodness, as illustrated in Fig. \ref{fig:ff_inference}.
	\end{remark}
	\begin{figure}[h]
		\centering
		\begin{tikzpicture}[node distance=0.6cm and 0.8cm, font=\scriptsize, >=Stealth]
			% Input
			\node (x) [draw, rectangle, minimum width=0.7cm, minimum height=0.2cm] {$\bm{x}$};
			
			% Augmented inputs
			\node (xy1) [draw, rectangle, right=of x, yshift=1cm, minimum width=1.4cm, minimum height=0.5cm] {$(\bm{x}, {y}_1)$};
			\node (xy2) [draw, rectangle, right=of x, minimum width=1.4cm, minimum height=0.5cm] {$(\bm{x}, {y}_2)$};
			\node (xyC) [draw, rectangle, right=of x, yshift=-1cm, minimum width=1.4cm, minimum height=0.5cm] {$(\bm{x}, {y}_C)$};
			
			% Dots
			\node (dots) [right=of x, xshift=0.4cm, yshift=-0.4cm] {$\vdots$};
			
			% Goodness blocks
			\node (g1) [draw, rectangle, right=of xy1, minimum width=1cm, minimum height=0.5cm] {$G_{\bm \theta}(\bm x, y_1)$};
			\node (g2) [draw, rectangle, right=of xy2, minimum width=1cm, minimum height=0.5cm] {$G_{\bm \theta}(\bm x, y_2)$};
			% Dots
			\node (dots) [right=of xy2, xshift=0.4cm, yshift=-0.4cm] {$\vdots$};
			\node (gC) [draw, rectangle, right=of xyC, minimum width=1cm, minimum height=0.5cm] {$G_{\bm \theta}(\bm x, y_C)$};
			
			% Selection
			\node (select) [draw, rectangle, right=of g2, xshift=0.01cm, minimum width=1.6cm, minimum height=0.5cm, font=\scriptsize] 
			{$\hat{y} = \arg\max\limits_{{y} \in \mathcal{Y}}\, G_{\bm \theta}(\bm{x}, {y})$
			};
			
			% Arrows
			\draw[->] (x.east) -- ++(0.2,0) |- (xy1.west);
			\draw[->] (x) -- (xy2);
			\draw[->] (x.east) -- ++(0.2,0) |- (xyC.west);
			
			\draw[->] (xy1) -- (g1);
			\draw[->] (xy2) -- (g2);
			\draw[->] (xyC) -- (gC);
			
			\draw[->] (g1) to[out=0, in=90] (select.north);
			\draw[->] (g2) -- (select);
			\draw[->] (gC) to[out=0, in=270] (select.south);
		\end{tikzpicture}
		\captionsetup{font=footnotesize}
		\caption{Inference process in the FF algorithm: the input $\bm{x}$ is paired with each candidate label ${y} \in \mathcal{Y}$ and the label with the highest goodness is selected.}
		\label{fig:ff_inference}
	\end{figure}
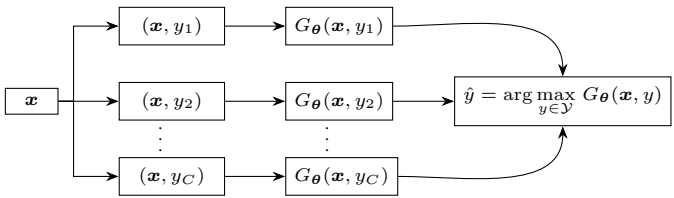
	
\vspace{-15 pt}
		\subsection{Related Work}
	The problem of adapting ML-based wireless algorithms after deployment under real-world practical constraints such as avoiding backpropagation has received limited attention in the literature. We review the closely-related prior works with focus on the link adaptation problem.
	\begin{enumerate}[leftmargin=*, itemsep=0pt, parsep=0pt]
		
		\item \textbf{Offline Training.} Several prior works adopt an offline training paradigm, where the ML model is trained once and its parameters remain fixed after deployment \cite{dong2018machine,van2021machine,ali2026online}. In this  approach, the primary objective is to construct a training dataset that broadly anticipates and represents all potential deployment scenarios which is an exhaustive process that is inherently time-consuming. Furthermore, despite this significant upfront complex training process, this static approach suffers from training-test mismatches, often resulting in significantly degraded inference performance within dynamic environments.
		
		\item \textbf{Offline Training with Online Calibration.} Several calibration techniques attempt to mitigate training-test mismatches by adjusting the model's output rather than its internal parameters \cite{baknina2020adaptive,huang2021deluxe,ali2026online}. These approaches rely on acknowledgment (ACK) and negative acknowledgment (NACK) of the physical downlink shared channel (PDSCH) transmissions. By tracking these signals, an empirical metric, such as the short-term BLER, is computed to dynamically adjust or scale the ML predictions. However, because the underlying ML model parameters remain frozen, these methods act as a superficial outdated correction rather than enabling the model to inherently adapt to the deployment environment.
		
		\item \textbf{Reinforcement Learning.} Although deep RL has been widely explored for continuous adaptation of ML-based wireless algorithms \cite{saxena2021reinforcement,an2024dragon,xu2024learning}, such solutions typically overlook the strict hardware limitations of wireless modems. Specifically, continuous RL training demands substantial memory for experience replay buffers, alongside significant computational power for ongoing backpropagation and gradient evaluations, making it largely prohibitive for resource-constrained devices. Furthermore, by operating purely on indirect reward signals, RL methodologies fail to exploit a critical structural advantage in cellular wireless systems: the inherent availability of delayed \emph{ground-truth labels} (e.g., ACKs/NACKs). Rather than relying on trial-and-error exploration, these delayed labels can be utilized to enable highly efficient, targeted supervised fine-tuning.
		
		\item \textbf{Table-based Baselines.} Prior to the proliferation of ML-driven techniques, table-based algorithms served as the de facto standard for wireless communications. Traditional link adaptation maps estimated channel conditions, such as signal-to-interference-plus-noise ratio (SINR) or mean mutual information per bit (MMIB), to a channel quality indicator (CQI) index using predefined lookup tables \cite{peralta2022outer}.\\ 
		To adapt to dynamic environments, OLLA algorithms utilize ACK/NACK feedback to continuously adjust an offset margin. While recent advancements, such as the exponential decay mechanism  \cite{mazumdar2026enhancing}, aim to accelerate OLLA convergence, a fundamental structural limitation remains. Specifically, conventional OLLA updates are typically isolated to the specific CQI currently scheduled for transmission. Because the remainder of the lookup table is not jointly updated, transitioning to a different CQI relies on stale, uncalibrated mappings. This severely bottlenecks the  adaptation process and frequently results in the selection of CQIs with high short-term BLER \cite{Ali2026lighttune}.

		\item \textbf{Parametric Adaptation.} Concurrent \cite{wiesmayr2025salad} and subsequent research \cite{maggi2026sinr} tackle the link adaptation problem by modeling the mapping between signal quality and BLER using differentiable sigmoid functions. These approaches propose an online convex optimization framework, implemented at the \emph{base station} which applies continuous gradient descent to minimize the cross-entropy loss between ACK/NACK feedback and predicted short-term BLER. While these methods provide smoother margin adjustments than conventional OLLA, they rely on one-dimensional parametric heuristics and require continuous gradient evaluations. Consequently, they remain agnostic to high-dimensional UE-side channel dynamics, such as Doppler and delay spreads, and lack the representational capacity required to proactively adapt to the complex non-stationarities of 6G networks.
		
\item \textbf{Digital Twins.} Recent advances advocate decentralizing DTs onto UE and edge devices to minimize latency and preserve privacy. A DT's lifecycle spans \textit{creation} (offline training), \textit{synchronization} (real-time adaptation), and \textit{retirement} (model replacement). While creation is well-studied, synchronization under strict UE compute and memory constraints remains an open problem, as existing techniques rely on computationally expensive methods such as RL \cite{tong2025continual}.
	\end{enumerate}
\textbf{Our Work.} \name{} addresses these limitations through a closed-loop online fine-tuning framework. In contrast to output-calibration methods, \name{} directly updates model parameters; in contrast to RL-based methods, it does so without backpropagation or a replay buffer. Updates are triggered only when the prediction error reaches or exceeds a threshold $\delta$, ensuring minimal computational overhead. \\When \name{} is deployed for on-device DT synchronization, this prediction error serves as the twin fidelity metric, quantifying the physical-digital alignment between the live channel and its parametric model, and $\delta$ defines the fidelity threshold below which synchronization is deemed acceptable. We further provide provable convergence guarantees for \name{} under training-test distribution shift.

	\section{Proposed Fine-tuning Algorithm: \name{}}
	\label{sec:proposed}
	In this section, we provide our  closed-loop  framework that uses the FF algorithm both to train a baseline ML model offline and to fine-tune it.
	\vspace{-10 pt}
	\subsection{Proposed Loss: Quadratic Softplus Approximation}
	The gradient computation for the  Softplus loss function requires evaluating exponentiations and divisions, which introduce undesirable computational latency and  complexity. To circumvent this overhead, we propose an alternative smooth quadratic loss function. 
	This quadratic loss encourages the goodness values to exceed the threshold $T$ for positive samples and fall below the threshold $T$ for negative samples. For layer $l$, this  loss is defined as
	\begin{align}
		\label{eqn:proposed_loss}
		\mathcal{L}_{\text{Prop},l} &= \frac{1}{|\mathcal{S}_{\text{+}}| M_l} \sum_{i \in \mathcal{S}_{\text{+}}} \sum_{j=1}^{M_l} \left( (g_{\text{+},l}^{(i)}[j] - T)^2 - 4(g_{\text{+},l}^{(i)}[j] - T) \right) \notag \\
		&\quad + \frac{1}{|\mathcal{S}_{\text{-}}| M_l} \sum_{i \in \mathcal{S}_{\text{-}}} \sum_{j=1}^{M_l} \left( (g_{\text{-},l}^{(i)}[j] - T)^2 + 4(g_{\text{-},l}^{(i)}[j] - T) \right).
	\end{align}
	We provide the derivation and experimental validation of the proposed loss function in Appendix \ref{app:loss_experiment}, demonstrating that it achieves accuracy comparable to the Softplus loss.\\ For the remainder of this paper, the proposed loss function for layer $l$ is denoted by $\mathcal{L}_l$ for simplicity.

	\subsection{Closed-Form Gradient Computations}
	To fine-tune the ML model using the proposed quadratic loss, we compute the local closed-form gradients for each layer $l$ of the multilayer perceptron (MLP). Let $\bm{W}_l$ and $\bm{b}_l$ denote the weight matrix and bias vector of layer $l$, respectively. We define the augmented parameter matrix $\bm{\Theta}_l = [\bm{W}_l, \bm{b}_l] \in \mathbb{R}^{M_l \times (M_{l-1} + 1)}$. This matrix formulation is equivalent to the flattened parameter vector $\bm{\theta}_l \in \mathbb{R}^{d_l}$ introduced previously, where the total number of layer parameters is $d_l := M_l(M_{l-1} + 1)$. Correspondingly, we define the augmented input activations for the positive and negative samples as $\tilde{\bm{h}}_{\text{+},l-1} = [\bm{h}_{\text{+},l-1}^\top, 1]^\top$ and $\tilde{\bm{h}}_{\text{-},l-1} = [\bm{h}_{\text{-},l-1}^\top, 1]^\top$, respectively. 
	
	In \name, each update step utilizes a single positive sample $\bm{z}_{\text{+}} = [\bm{x}^\top, y_{\text{+}}]^\top$ and a single negative sample $\bm{z}_{\text{-}} = [\bm{x}^\top, y_{\text{-}}]^\top$, where $y_{\text{-}} \in \mathcal{Y} \setminus \{ y_{\text{+}} \}$. Consequently,  $|\mathcal{S}_{\text{+}}| = |\mathcal{S}_{\text{-}}| = 1$. The derivative of the proposed layer-wise loss with respect to the augmented parameter matrix is computed as
	\begin{align}
		\nabla_{\bm{\Theta}_l} \mathcal{L}_{l} = \nabla_{\bm{\Theta}_l} \mathcal{L}_{\text{+}, l} + \nabla_{\bm{\Theta}_l} \mathcal{L}_{\text{-}, l},
	\end{align}
	where $\nabla_{\bm{\Theta}_l} \mathcal{L}_{\text{+}, l}$ and $\nabla_{\bm{\Theta}_l} \mathcal{L}_{\text{-}, l}$ denote the gradient  corresponding to the positive and negative sample, respectively. Assuming the use of rectified linear unit (ReLU) activations, these gradients can be expressed in a closed-form matrix notation as
	\begin{align}
		\label{eqn:gradients}
		\nabla_{\bm{\Theta}_l} \mathcal{L}_{\text{+}, l} &= \frac{4}{M_l} \big[ (\bm{g}_{\text{+},l} - T - 2) \odot \bm{h}_{\text{+},l} \odot \mathds{1}(\bm{h}_{\text{+},l} > 0) \big] \tilde{\bm{h}}_{\text{+},l-1}^\top, \notag \\
		\nabla_{\bm{\Theta}_l} \mathcal{L}_{\text{-}, l} &= \frac{4}{M_l} \big[ (\bm{g}_{\text{-},l} - T + 2) \odot \bm{h}_{\text{-},l} \odot \mathds{1}(\bm{h}_{\text{-},l} > 0) \big] \tilde{\bm{h}}_{\text{-},l-1}^\top,
	\end{align}
	where $\odot$ denotes the element-wise Hadamard product and $\mathds{1}(\cdot)$ is the indicator function applied element-wise. 
	Since the  fine-tuning process utilizes exactly one positive and one negative sample,  the gradient computation and the parameter update are performed only once. Crucially, each parameter update is purely local to its respective layer. By requiring only the current layer's activations and the forward-propagated inputs from the preceding layer, this localized approach eliminates the memory and computational overhead associated with BP.
	
	\subsection{Threshold-based Fine-tuning Without Replay Buffer}
	To  eliminate the memory overhead associated with experience replay buffers at the UE, \name{} employs a threshold-based update policy, enabling the model to be fine-tuned in an opportunistic, sample-by-sample manner. The \textit{surprise} of an observation at time $t$ is quantified by the prediction error, denoted as $e(\hat{y}^{(t)}, y_{\text{+}}^{(t)})$, between the model's predicted label $\hat{y}^{(t)}$ and the delayed ground-truth label $y_{\text{+}}^{(t)}$. When this error meets or exceeds a predefined threshold $\delta$, the model identifies a distribution shift and triggers a local update. For this update, the true observation $\bm{z}_{\text{+}}^{(t)} = [{\bm{x}^{(t)}}^\top, y_{\text{+}}^{(t)}]^\top$ is utilized as the positive training pair. A corresponding negative sample $\bm{z}_{\text{-}}^{(t)} = [{\bm{x}^{(t)}}^\top, y_{\text{-}}^{(t)}]^\top$ is constructed by pairing the current input feature vector $\bm{x}^{(t)}$ with an incorrect label $y_{\text{-}}^{(t)} \in \mathcal{Y} \setminus \{ y_{\text{+}}^{(t)} \}$. This on-the-fly fine-tuning mechanism is detailed in Algorithm \ref{alg:fine_tuning} and  Fig. \ref{fig:finetune}. We consider two negative sampling strategies as follows. \\ 1) \textbf{Uniform Negative Sampling}. $y_{\text -}^{(t)}$ is sampled uniformly at random from the set of incorrect labels $\mathcal{Y} \setminus \{y_{\text +}^{(t)}\}$. This approach is widely adopted in the FF literature \cite{hinton2022forward,torres2025advancements}.\\ 2) \textbf{Hard Negative Sampling}. The negative label is set to the model's current erroneous prediction: $y_{\text -}^{(t)} = \hat{y}^{(t)}$. Inspired by hard triplet mining and adversarial training \cite{schroff2015facenet}, this selection forces the model to specifically penalize its most confident errors, thereby sharpening the decision boundaries more effectively during fine-tuning.
	
	\begin{algorithm}
		\caption{Proposed Fine-tuning Algorithm: \name{}}
		\label{alg:fine_tuning}
		\begin{algorithmic}[1]
			\small 
			\REQUIRE Current model parameters $\bm{\theta}^{(t)}$, feature vector  $\bm{x}^{(t)}$, delayed actual metric $y_{\text{+}}^{(t)}$, fine-tuning threshold $\delta$, Adam optimizer parameters ($t_{\text A}, \beta_1, \beta_2, \epsilon$), fine-tuning learning rate $\alpha_f$ 
			\ENSURE ML prediction $\hat y^{(t)}$, fine-tuned model parameters $\bm{\theta}^{(t+1)}$
			\STATE Compute model prediction: $\hat{y}^{(t)} = \arg\max_{{y} \in \mathcal Y} G_{\bm \theta^{(t)}}\left(\bm{x}^{(t)}, {y} \right)$
			\IF{\(e(\hat{y}^{(t)}, y_{\text{+}}^{(t)}) \geq \delta\)}
			\STATE Construct positive sample: $\bm{z}_{\text{+}}^{(t)} = [{\bm{x}^{(t)}}^\top, y_{\text{+}}^{(t)}]^\top$ \vspace{1pt}
			\STATE Construct negative sample: $\bm{z}_{\text{-}}^{(t)} = [{\bm{x}^{(t)}}^\top, y_{\text{-}}^{(t)}]^\top$
			\STATE Compute gradients $\nabla_{\bm{\theta}} \mathcal{L}^{(t)}$ using $\bm{z}_{\text{+}}^{(t)}$ and $\bm{z}_{\text{-}}^{(t)}$ (as in \eqref{eqn:gradients}) 
			
			\IF{Standard Adam variant}
			\STATE $t_{\text A} \gets t_{\text A} + 1$
			\ELSIF{One-step variant}
			\STATE $t_{\text A} \gets 1, \quad \bm{m}_{0} \gets \bm{0}, \quad \bm{v}_{0} \gets \bm{0}$
			\ENDIF
			
			\STATE Compute updated model parameters as follows \cite{kingma2015adam}: 
			\begin{align}
				&\bm{m}_{t_{\text A}} = \beta_1 \bm{m}_{t_{\text A}-1} + (1 - \beta_1) \nabla_{\bm{\theta}} \mathcal{L}^{(t)}, \\
				&\bm{v}_{t_{\text A}} = \beta_2 \bm{v}_{t_{\text A}-1} + (1 - \beta_2) (\nabla_{\bm{\theta}} \mathcal{L}^{(t)})^2, \\
				&\Delta \bm{\theta}^{(t)} = \frac{\bm{m}_{t_{\text A}} / (1 - \beta_1^{t_{\text A}})}{\sqrt{\bm{v}_{t_{\text A}} / (1 - \beta_2^{t_{\text A}})} + \epsilon}, \\
				&\bm{\theta}^{(t+1)} \gets \bm{\theta}^{(t)} - \alpha_f \Delta \bm{\theta}^{(t)}.
			\end{align}
			\ELSE
			\STATE $\bm{\theta}^{(t+1)} \gets \bm{\theta}^{(t)}$
			\ENDIF
		\end{algorithmic}
	\end{algorithm}
	
	\begin{figure}
		\captionsetup{font=footnotesize}
		\centering
		\resizebox{\columnwidth}{!}{%
			\hspace{1cm}
			\begin{tikzpicture}[>=Stealth, thick, scale=0.72,font=\large]
				\node (x) at (0,2.9) {\( \bm{x}^{(t)} \)};
				\node[draw, rounded corners, fill=gray!10, text width=3.5 cm, align=center, minimum height=1.2cm] 
				(model) at (0,0.7) {ML Model\\\( \bm{\theta} \)};
				\node[draw, rounded corners, fill=gray!10, text width=3.5cm, align=center, minimum height=1.2cm] 
				(error) at (0,-2.4) {Error Computation\\\( e(\hat{y}^{(t)}, y_{\text{+}}^{(t)}) \)};
				\node (y) at (4.2,-2.4) {\( y_{\text{+}}^{(t)} \)};
				\node[draw, rounded corners, fill=gray!10, text width=5.5cm, align=center, minimum height=1.2cm] 
				(tune) at (-12,-2.4) {Compute Gradients};
				\node[draw, rounded corners, fill=gray!10, text width=5.5cm, align=center, minimum height=1.2cm] 
				(update) at (-12,0.7) {Parameter Update\\ $\bm{\theta}^{(t+1)} = \bm{\theta}^{(t)} - \alpha_f \Delta \bm \theta^{(t)}$};
				\draw[->, shorten >=2pt] (x) -- (model);
				\draw[->, shorten >=2pt] (model) -- node[pos=0.5, right] {\( \hat{y}^{(t)} = \arg\max_{{y} \in \mathcal Y} G(\bm{x}^{(t)}, y) \)} (error);
				\draw[->, shorten >=2pt] (y) -- (error);
				\draw[->, shorten >=2pt] (error) -- node[pos=0.5, above, yshift=0pt] {\( e(\hat{y}^{(t)}, y_{\text{+}}^{(t)}) \geq \delta \)} (tune);
				\draw[->, shorten >=2pt] (tune) -- node[pos=0.5, above, xshift=-25pt, yshift=-8pt] {\( \Delta \bm \theta^{(t)} \)} (update);
				\draw[->, shorten >=2pt] (update) -- node[pos=0.3, above, xshift=10pt,yshift=1pt] {\( \bm \theta \leftarrow \bm{\theta}^{(t+1)} \)} (model);
			\end{tikzpicture}
		}
		\caption{\name{} uses the delayed true label $y_{\text{+}}^{(t)}$ to compute the prediction error and fine-tune the model if needed.}
		\label{fig:finetune}
	\end{figure}
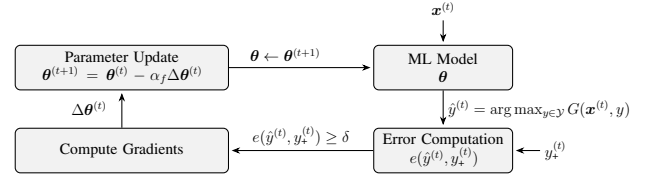

	\subsection{Fine-tuning Variants}
	Although \name{} supports SGD and other optimizers, our focus is on variants of Adam optimizer \cite{kingma2015adam} and normalized SGD. Adam dynamically adapts parameters using two running averages: the first moment (gradient mean) $\bm{m}_t$ and the second moment (gradient variance) $\bm{v}_t$.\\
	1) \textbf{Standard Adam Update}: The optimizer maintains an internal counter $t_{\text A}$ that exclusively tracks the number of executed \name{} fine-tuning updates. Specifically, $t_{\text A}$ is incremented only when the fine-tuning condition $e^{(t)} \ge \delta$ is triggered. The updates reuse the moments from the previous fine-tuning instance, following the standard recursive form $\bm{m}_{t_{\text A}} = \beta_1 \bm{m}_{t_{\text A}-1} + (1 - \beta_1) \nabla_{\bm{\theta}} \mathcal{L}^{(t)}$ and $\bm{v}_{t_{\text A}} = \beta_2 \bm{v}_{t_{\text A}-1} + (1 - \beta_2) (\nabla_{\bm{\theta}} \mathcal{L}^{(t)})^2$, where $\beta_1$ and $\beta_2$ are the exponential decay rates. \\
	2) \textbf{One-step Update}: To  eliminate the need for storing previous moments, this variant resets the internal counter to $t_{\text A} = 1$ at every fine-tuning instance. Consequently, historical moments are effectively reinitialized to zero, and the update simplifies to $\bm{m}_{1} = (1 - \beta_1) \nabla_{\bm{\theta}} \mathcal{L}^{(t)}$ and $\bm{v}_{1} = (1 - \beta_2) (\nabla_{\bm{\theta}} \mathcal{L}^{(t)})^2$. This variant is  memory-efficient as the UE is not required to store the high-dimensional moments across the fine-tuning steps.
	
	\begin{remark}[\textbf{Rationale for One-step Update}]
		In the online, sparse-update setting of \name{}, ``surprising'' samples frequently indicate a fundamental distribution shift in the underlying wireless channel. By resetting the internal Adam moments, this variant ensures that the parameter update relies strictly on the current gradient. Consequently, the update simplifies to a normalized gradient step given by (when $\epsilon = 0, \beta_1 = 0, \beta_2 = 0$)
		\begin{equation}
			\label{eqn:sign_update}
			\bm{\theta}^{(t+1)} \gets \bm{\theta}^{(t)} - \alpha_f \cdot \text{sgn}\left( \nabla_{\bm{\theta}} \mathcal{L}^{(t)} \right),
		\end{equation}
		where $\text{sgn}(\cdot)$ denotes the element-wise sign function. This gradient update is also advantageous for a resource-limited UE, as it replaces expensive divisions with sign extractions. 
	\end{remark}

	\section{Convergence of \name{}}
	\label{sec:convergence}
	We  provide the convergence guarantee for \name{} under SGD, showing that
	the average frequency of prediction errors reaching or exceeding any fixed threshold $\delta > 0$ converges to $0$ as the number of \name{} updates increases, despite the training-test distribution mismatch. The intuition is that whenever the prediction error is large, the model update reduces the loss by a guaranteed amount. Consequently, large errors can occur only finitely many times, and their frequency must decay to $0$.
	We start with definitions and assumptions. 

	\subsection{Definitions and Assumptions}
	
	\begin{definition}[\textbf{Error-driven Update Indicator}]
		For an error tolerance $\delta > 0$, the update at step $t$ is controlled by a binary random variable $I^{(t)}_{\delta}$, defined as
		\begin{equation}
			I^{(t)}_\delta := \mathds{1}\{e^{(t)} \ge \delta\},
		\end{equation} 
		where $e^{(t)} = e^{(t)}(\hat{y}^{(t)}, y^{(t)}_{\text{+}})$ denotes the prediction error.
	\end{definition}
	\noindent Hence, $\mathbb{E}[I^{(t)}_\delta] = \Pr(e^{(t)} \geq \delta)$. 
	
	\begin{definition}[\textbf{Filtration}]
		The filtration $\mathcal{F}_t$ represents the history of the stochastic process up to time $t$ and is defined as the $\sigma$-algebra generated by the offline model parameters and the sequence of positive samples observed up to time $t$: 
		$
		\mathcal{F}_t = \sigma\bigl(\bm{\theta}^{(1)}, \bm{z}^{(1)}_{\text +}, \dots, \bm{z}^{(t)}_{\text +}\bigr).
		$
	\end{definition}

	\begin{assumption}[\textbf{ReLU MLPs}]
		\label{assumption:ReLUs}
		The underlying ML model of \name{} is assumed to be a ReLU MLP. 
	\end{assumption}
	
	\begin{assumption}[\textbf{Bounded Input and Model}]
		\label{assumption:boundedness}
		There exist positive constants $B_z, B_\theta$ such that for all time steps $t$ and all layers $l$: (i) $\|\bm{z}^{(t)}\|_2 \leq B_z$, where $\bm{z}^{(t)} = [{\bm{x}^{(t)}}^\top, y^{(t)}]^\top \in \mathbb{R}^{d+1}$, and (ii) $\|\bm{\theta}_{l,j}^{(t)}\|_2 \leq B_\theta$ for all neurons $j$ in layer $l$, where $\bm{\theta}_{l,j}^{(t)}$ denote the parameters of neuron $j$ in layer $l$.
	\end{assumption}
	
	Next, to model the training-test mismatch problem, we consider Assumption \ref{assumption:distributions}.
	
	\begin{assumption}[\textbf{Training and Online Data Distributions}]
		\label{assumption:distributions}
		The offline training data are i.i.d. from distribution $\mathcal{D}_1$, and the online fine-tuning data $\{\bm{z}^{(t)}\}_{t\ge 1}$ are i.i.d. from a (possibly different) distribution $\mathcal{D}_2$.
	\end{assumption}
	\begin{assumption}[\textbf{Gradient Lower Bound}]
		\label{assumption:gradient_lower_bound}
		For any fixed $\delta > 0$, there exists a constant $\gamma_2(\delta) > 0$ such that whenever $e^{(t)} \geq \delta$, the expected squared norm of the gradient (with respect to the random negative sample) under $\mathcal{D}_2$ satisfies\footnote{The subscript in $\gamma_2$ means that this is under the distribution $\mathcal{D}_2$.}
		\begin{equation}
			\mathbb{E}_{y^{(t)}_{\text{-}}}\!\left[\|\nabla_{\bm{\theta}_L} \mathcal{L}_L^{(t)}(\bm{\theta}_L^{(t)})\|_2^2 \;\big|\; \mathcal{F}_t,\; e^{(t)}\geq\delta\right] \geq \gamma_2(\delta).
		\end{equation}
	\end{assumption}
	\noindent This assumption is analogous to the gradient dominance or Polyak--\L{}ojasiewicz (PL) conditions frequently invoked in non-convex optimization literature \cite{karimi2016linear}. It ensures that the ``learning signal'' remains strictly bounded away from zero whenever the model's performance is not as desired (i.e., $e^{(t)} \geq \delta$). In effect, this condition guarantees that the stochastic nature of the negative label selection does not lead to vanishing gradients during the online fine-tuning process. We provide an informal intuitive justification for this assumption in Appendix \ref{app:convergence}. 
	
	\subsection{Preliminary Lemmas}
	We  present  foundational lemmas that we build on to prove our convergence theorem. The  proofs are deferred to Appendix~\ref{app:convergence}. We first establish the boundedness of activations, gradients and loss followed by the  smoothness of the loss.
	
	\begin{lemma}[\textbf{Bounded Activations}]
		\label{lemma:bounded_activations}
		Under Assumption~\ref{assumption:boundedness} (i.e., bounded input and model), there exists a positive constant $B_h$ such that for all time steps $t$ and all layers $l$, $\|\bm{h}_l^{(t)}\|_2 \le B_h$. In particular, one can take $B_h = \max_{0\le l\le L} B_l$, where $B_0 = B_z$ and for $l \ge 1$, $B_l = \sqrt{M_l}\, B_\theta\,(B_{l-1}+1).$
	\end{lemma}
	
	\begin{lemma}[\textbf{Bounded Gradient}]
		\label{lemma:bounded_gradient}
		Under Assumption \ref{assumption:boundedness} (i.e., bounded input and model), for all time steps $t$, layers $l$, and neurons $j$, the gradient of the layer loss with respect to the parameters of neuron $j$ is bounded by
		\begin{equation}
			\|\nabla_{\bm{\theta}_{l,j}} \mathcal{L}_l^{(t)}(\bm{\theta}_l^{(t)})\|_2 \leq \frac{8(B_h^2 + T + 2) B_h (B_h + 1)}{M_l},
		\end{equation}
		where $\nabla_{\bm{\theta}_{l,j}} \mathcal{L}_l^{(t)}(\bm{\theta}_l)$ denotes the partial derivative of $\mathcal{L}_l^{(t)}$ with respect to $\bm{\theta}_{l,j}$.
	\end{lemma}
	
	\begin{lemma}[\textbf{Bounded Loss}]
		\label{lemma:bounded_loss}
		Under Assumption \ref{assumption:boundedness} (i.e., bounded input and model), there exists a constant $M > 0$ such that the loss is bounded as
		\begin{equation}
			\sup_{t \ge 1} \left| \mathcal{L}_{L}^{(t)}(\bm{\theta}_L^{(t)}) \right| \le M,
		\end{equation} 
		where $M = (B_h^2 + T + 2)^2 + 4(B_h^2 + T + 2).$
	\end{lemma}
	
	\begin{lemma}[\textbf{Smoothness}]
		\label{lemma:smoothness}
		Under Assumption \ref{assumption:boundedness} (i.e., bounded input and model), for each layer $l$, the loss function $\mathcal{L}_l^{(t)}(\bm{\theta})$ is $\rho_l$-smooth with 
		$
		\rho_l = \frac{8\bigl(3B_h^2 + T + 2\bigr)(B_h + 1)^2}{M_l}.
		$
		That is, for any $\bm{\theta}_l, \bm{\theta}_l' \in \mathbb{R}^{d_l}$, 
		$
		\|\nabla \mathcal{L}_l^{(t)}(\bm{\theta}_l) - \nabla \mathcal{L}_l^{(t)}(\bm{\theta}'_l)\|_2 \le \rho_l \, \|\bm{\theta}_l - \bm{\theta}'_l\|_2.
		$
	\end{lemma}
	
	\noindent Next, we consider the descent lemma \cite[Lemma 1.2.3]{nesterov2013introductory}.
	\begin{lemma}[\textbf{Descent Lemma}]
		\label{lemma:descent}
		For any $\rho_l$-smooth function, we have
		\begin{equation}
			\mathcal{L}_l^{(t)}(\bm{\theta}'_l) \le \mathcal{L}_l^{(t)}(\bm{\theta}_l) + \nabla\mathcal{L}_l^{(t)}(\bm{\theta}_l)^\top (\bm{\theta}'_l-\bm{\theta}_l) + \frac{\rho_l}{2}\|\bm{\theta}'_l-\bm{\theta}_l\|_2^2.
		\end{equation}
	\end{lemma}
	
	\noindent Finally, to handle the training-test distribution mismatch, we leverage Pinsker's inequality.
	\begin{lemma}[\textbf{Pinsker's Inequality} \cite{pinsker1964information}]
		\label{lemma:pinsker}
		Let $P$ and $Q$ be two probability distributions on a measurable space. For any measurable function $f$ with $\|f\|_\infty = \sup |f| \le M$, we have
		\begin{equation}
			|\mathbb{E}_P[f] - \mathbb{E}_Q[f]| \le M \sqrt{\frac{1}{2} D_{\mathrm{KL}}(P\|Q)},
		\end{equation}
		    where $D_{\mathrm{KL}}(P\|Q)$ denotes the Kullback--Leibler divergence of $P$ with respect to $Q$.
	\end{lemma}

	\subsection{Convergence Theorem}
	The proof follows the following key steps. 1) \textbf{Local decrease}: When a large error occurs (i.e., $I^{(t)}_\delta = 1$), the gradient step reduces the loss by at least $\frac{\alpha_f}{2}\|\nabla\mathcal{L}_L^{(t)}\|^2$. 2) \textbf{Gradient lower bound}: Assumption \ref{assumption:gradient_lower_bound} guarantees that whenever the error is large, the expected squared gradient norm is at least $\gamma_2(\delta) > 0$. Thus, each large error produces an expected decrease of at least $\frac{\alpha_f\gamma_2(\delta)}{2}$. 3) \textbf{Telescoping sum}: Summing these expected decreases over time and using stationarity of the data distribution yields the upper bound $\frac{\alpha_f\gamma_2(\delta)}{2}\sum_{t=1}^N \mathbb{E}[I^{(t)}_\delta] \le \mathbb{E}[\mathcal{L}_L^{(1)}(\bm{\theta}_L^{(1)})] - \mathcal{L}_L^*$, where $\mathcal{L}_L^* = \inf_{\bm{\theta}} \mathbb{E}_{\mathcal{D}_2}[\mathcal{L}_L^{(t)}(\bm{\theta})]$. 4) \textbf{Consequences}: Rearranging gives $\frac{1}{N}\sum_{t=1}^N \Pr(e^{(t)}\ge\delta) \le O(1/N)$ when the offline and online distributions are identical. 5) \textbf{Distribution shift}: When $\mathcal{D}_1 \neq \mathcal{D}_2$, Pinsker's inequality introduces an extra term $\sqrt{2 D_{\mathrm{KL}}(\mathcal{D}_2\|\mathcal{D}_1)/N}$.
	
	\begin{theorem}[\textbf{Convergence under Distribution Shift}]
		\label{thm:convergence}
		Suppose Assumptions \ref{assumption:ReLUs}--\ref{assumption:gradient_lower_bound} hold. For a fixed error tolerance $\delta > 0$ and a learning rate $\alpha_f \in (0, 1/\rho_L)$, \name{} with SGD satisfies for any $N \ge 1$:
		\begin{align}
			\frac{1}{N}\sum_{t=1}^N \Pr_{\hspace{5pt}\mathcal{D}_2} \bigl(e^{(t)} \ge \delta\bigr) &\le \frac{2 \left[ \mathbb{E}_{\mathcal{D}_1}[\mathcal{L}_L^{(1)}(\bm{\theta}_L^{(1)})] - \mathcal{L}_L^* \right]}{\alpha_f \gamma_2(\delta) N} \notag \\
			&\quad + \frac{2M}{\alpha_f \gamma_2(\delta)} \sqrt{\frac{2 D_{\mathrm{KL}}(\mathcal{D}_2 \| \mathcal{D}_1)}{N}},
		\end{align}
		where $\mathcal{L}_L^* = \inf_{\bm{\theta}} \mathbb{E}_{\mathcal{D}_2}[\mathcal{L}_L^{(t)}(\bm{\theta})]$ is the minimum achievable expected loss under $\mathcal{D}_2$, the constant $M$ satisfies $\sup_{t \ge 1} \lvert \mathcal{L}_{L}^{(t)}(\bm{\theta}_L^{(t)}) \rvert \le M$, and $\gamma_2(\delta)$ is the gradient lower bound from Assumption \ref{assumption:gradient_lower_bound}.
	\end{theorem}
	
	\begin{corollary}[Asymptotic Convergence of Average Error Probability]
		\label{corollary:asymptotic_average}
		Under the conditions of Theorem \ref{thm:convergence}, for any fixed error tolerance $\delta > 0$, the average probability of a significant prediction error satisfies
		\begin{equation}
			\lim_{N \to \infty} \frac{1}{N} \sum_{t=1}^N \Pr_{\hspace{5pt}\mathcal{D}_2} \bigl(e^{(t)} \ge \delta\bigr) = 0.
		\end{equation}
	\end{corollary}
	\begin{corollary}[\textbf{Pointwise Convergence under Identical Distributions}]
		\label{corollary:pointwise_identical}
		If $\mathcal{D}_1 = \mathcal{D}_2$, under the conditions of Theorem \ref{thm:convergence}, we have
		\begin{equation}
			\lim_{t\to\infty} \Pr\bigl(e^{(t)} \ge \delta\bigr) = 0.
		\end{equation}
	\end{corollary}
	
	\section{Applications of \name{}}
	\label{sec:app}
To demonstrate the efficacy of \name{}, we instantiate it within a 6G link adaptation scenario, where the UE maintains a short-term BLER predictive model and \name{} continuously refines it using delayed ACK/NACK feedback. 
	
	\subsection{BLER Prediction}
	\label{sec:bler_pred}
	We consider the downlink short-term BLER prediction problem. In this context, ``short-term'' refers to a prediction horizon spanning several PDSCH transmission slots corresponding to the periodicity of the CSI reports. Such short-term BLER prediction at the UE side is essential in  cellular systems (e.g., 5G/6G) for enabling accurate and timely CSI reporting.

 The CSI reference signal (CSI-RS) period denotes the interval between consecutive downlink reference signals transmitted by the next-generation NodeB (gNB) for channel acquisition. Based on these estimates, the UE subsequently generates and transmits CSI reports back to the gNB. CSI-RS transmissions can be configured as either periodic or aperiodic. In periodic mode, reference signals are transmitted at regular intervals. In the aperiodic mode, however, the interval between reference signals is variable.
	
	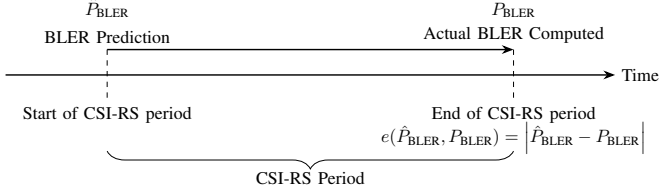
\begin{figure}
		\centering
		\resizebox{\columnwidth}{!}{%
			\begin{tikzpicture}[>=Stealth, node distance=2.5cm]
				\draw[->, thick] (0,0) -- (12,0) node[right] {Time};
				\draw[dashed] (2,0.5) -- (2,-0.5) node[below] {Start of CSI-RS period};
				\draw[dashed] (10,0.5) -- (10,-0.5) node[below] {End of CSI-RS period};
				\node[above] at (2,0.5) {BLER Prediction};
				\node[above=1cm] at (2,0.01) {\( \hat{P}_{\text{BLER}} \)};
				\node[above] at (10,0.5) {Actual BLER Computed};
				\node[above=1cm] at (10,0.01) {\( P_{\text{BLER}} \)};
				\draw[decorate,decoration={brace,mirror,amplitude=10pt}] (2,-1.5) -- (10,-1.5) node[midway,below=8pt]{CSI-RS Period};
				\draw[->, thick] (2,0.5) -- (10,0.5);
				\node[below=0.8cm] at (10,0) {\( e(\hat{P}_{\text{BLER}}, P_{\text{BLER}}) = \left| \hat{P}_{\text{BLER}} - P_{\text{BLER}} \right| \)};
			\end{tikzpicture}
		}
		\captionsetup{font=footnotesize}
		\caption{Timeline showing BLER prediction at the start and actual BLER computation at the end of a CSI-RS period}
		\label{fig:bler_timeline}
	\end{figure}
	
	\begin{algorithm}[htb!]
		\small 
		\caption{BLER Prediction Algorithm: \textsc{BLER-Predict}}
		\label{alg:bler_predict}
		\begin{algorithmic}[1]
			\REQUIRE Current parameters $\bm{\theta}^{(t)}$, feature vector $\bm{x}^{(t)}$, quantized BLER classes $\mathcal{P} = \{p_1, p_2, \cdots, p_C \}$
			\ENSURE Predicted short-term BLER $\hat{P}_{\text{BLER}}$
			
			\STATE $G_{\max} \gets -1$, $\hat{P}_{\text{BLER}} \gets 0$
			
			\FOR{\textbf{each} candidate BLER $p \in \mathcal{P}$}
			\STATE Compute the goodness $G = G_{\bm \theta^{(t)}}(\bm x^{(t)}, p)$
			
			\IF{$G > G_{\max}$}
			\STATE $G_{\max} \gets G$, $\hat{P}_{\text{BLER}} \gets p$ 
			\ENDIF
			\ENDFOR
			
			\STATE \textbf{return} $\hat{P}_{\text{BLER}}$
		\end{algorithmic}
	\end{algorithm}
	
	The inference mechanism for estimating the short-term BLER is detailed in Algorithm \ref{alg:bler_predict}. At the beginning of the $t$-th CSI-RS period, the ML model predicts the short-term BLER, denoted as $\hat{P}_{\text{BLER}}^{(t)}$, utilizing input features such as the CSI-RS signal-to-noise ratio (SNR). Following the PDSCH transmissions within this CSI-RS period, the true BLER at the end of the period is computed as 
	\begin{align}
		P_{\text{BLER}}^{(t)} = \frac{n_\text{NACK}^{(t)}}{n_\text{NACK}^{(t)} + n_\text{ACK}^{(t)}}, 
	\end{align}
	where $n_\text{ACK}^{(t)}$ and $n_\text{NACK}^{(t)}$ denote the total number of PDSCH ACKs and NACKs recorded during the $t$-th CSI-RS period, respectively. The UE then evaluates the   prediction error as 
	\begin{align}
		e^{(t)}\left(\hat{P}_{\text{BLER}}^{(t)}, P_{\text{BLER}}^{(t)}\right) = \left| \hat{P}_{\text{BLER}}^{(t)} - P_{\text{BLER}}^{(t)} \right|.
	\end{align} 
	Within the overarching link adaptation framework, the opportunistic fine-tuning process is triggered if and only if this prediction error meets or exceeds a predefined threshold $\delta$. This condition ensures that the model undergoes parameter updates exclusively when its predictions deviate significantly from the empirical ground truth, enabling  targeted adaptation.
	
	\begin{remark}[\textbf{Using the FF Algorithm for Regression}]
		The FF algorithm was originally designed for classification tasks. To adapt it for regression, we discretize the continuous target variable (i.e., BLER) into a finite set of discrete classes $\mathcal{P} = \{p_1, p_2, \dots, p_C\}$. This introduces a trade-off between the prediction accuracy and the inference latency. However, as discussed in Sec. \ref{sec:CQI}, our link adaptation framework only requires a coarse estimate to reliably determine whether the short-term BLER exceeds a predefined threshold.
	\end{remark}
	
	\subsection{Channel Quality Indicator (CQI) Selection}
	\label{sec:CQI}
	In 5G, the UE reports a CQI and a rank indicator (RI) to the gNB, which reflect the perceived downlink channel conditions. Based on the reported CQI, the gNB  selects an appropriate modulation and coding scheme (MCS) to optimize throughput while maintaining reliable communication.
	
	Traditional $\text{CQI}$ selection methods rely on look-up tables and may select a $\text{CQI}$ index that results in an excessively high BLER, particularly initially when the algorithm has not yet converged. 
	To mitigate this high short-term BLER problem, we propose an adaptive, ML-guided back-off strategy driven by the online BLER prediction framework as detailed in Algorithm \ref{alg:cqi_backoff}. Specifically, if the CQI selected by the conventional table-based algorithm yields a predicted BLER that meets or exceeds a predefined BLER threshold $\tau_{\text{BLER}}$ (e.g., $0.9$), the algorithm iteratively decrements the CQI index until the predicted short-term BLER falls below this threshold provided that the CQI does not drop below a predefined minimum allowable bound, $\text{CQI}_{\text{min}}$. Crucially, the candidate CQI is explicitly embedded as a feature within the input feature vector $\bm{x}^{(t)}$. This allows \BLER{} to predict the BLER of that CQI. 
	
	\begin{algorithm}[htb!]
		\small
		\caption{$\text{CQI}$ Back-off Algorithm: \CQI{}}
		\label{alg:cqi_backoff}
		\begin{algorithmic}[1]
			\REQUIRE Current model parameters $\bm{\theta}^{(t)}$; table-based CQI index $\text{CQI}^{(r)}$ and RI $r$; input feature vector $\bm{x}^{(t)}$; minimum allowable CQI index $\text{CQI}_{\text{min}}$; BLER target threshold $\tau_{\text{BLER}}$;  fine-tuning error threshold $\delta$; learning rate $\alpha_f$; BLER classes $\mathcal P$.
			\ENSURE Adjusted CQI: $\text{CQI}_{\text L}^{(r)}$, updated model parameters $\bm{\theta}^{(t+1)}$
			
			\STATE $\text{CQI}_{\text L}^{(r)} \gets \text{CQI}^{(r)}$
			
			\WHILE{\( \text{CQI}_{\text L}^{(r)} > \text{CQI}_{\text{min}} \)}
			\STATE Update the CQI feature in $\bm{x}^{(t)}$ to $\text{CQI}_{\text L}^{(r)}$
			\STATE $\hat{p} \gets \BLER(\bm{\theta}^{(t)}, \bm{x}^{(t)}, \mathcal P)$
			\IF{$ \hat{p} < \tau_{\text{BLER}}$}
			\STATE \textbf{break} 
			\ENDIF
			\STATE \( \text{CQI}_{\text L}^{(r)} \gets \text{CQI}_{\text L}^{(r)} - 1 \) 
			\ENDWHILE
			
			\STATE Update the CQI feature in $\bm{x}^{(t)}$ to $\text{CQI}_{\text L}^{(r)}$ 
			
			\STATE Transmit PDSCH utilizing the final adjusted $\text{CQI}_{\text L}^{(r)}$
			\STATE Compute $n_{\text{ACK}}^{(t)}$ and $n_{\text{NACK}}^{(t)}$ at the end of the CSI-RS Period 
			\STATE Calculate the empirical ground-truth BLER:
			$P_{\text{BLER}}^{(t)} \gets \frac{n_{\text{NACK}}^{(t)}}{n_{\text{NACK}}^{(t)} + n_{\text{ACK}}^{(t)}}$
			
			\STATE $\hat{P}_{\text{BLER}}^{(t)}, \bm{\theta}^{(t+1)} \gets \name\left(\bm{\theta}^{(t)}, \bm{x}^{(t)}, P_{\text{BLER}}^{(t)}, \delta, \alpha_f \right)$
			
			\STATE \textbf{return} $\text{CQI}_{\text L}^{(r)}, \bm{\theta}^{(t+1)}$
		\end{algorithmic}
	\end{algorithm}
	
	It is important to clarify that the loop in Alg.~\ref{alg:cqi_backoff} operates within a single CSI-RS reporting period. That is, the $\text{CQI}$ is reported at the end of the loop. We also note that \CQI{} does not change the RI selected by the baseline algorithm. 
	
	\begin{remark}[\textbf{False Alarm and Missed Detection Effect on CQI Selection.}]
		\label{remark:FA_MD}
		\CQI{} only requires a coarse estimate of the BLER. Specifically, the goal of the underlying BLER prediction algorithm in \CQI{} is to determine if the short-term BLER is below the threshold $\tau_{\text{BLER}}$ or not. Hence, it is critical to analyze the false alarm (FA) and the missed detection (MD) probabilities of the underlying BLER prediction algorithm denoted as $P_{\text{FA}}$ and $P_{\text{MD}}$, respectively. The FA and MD events are defined as $\mathcal{E}_{\text{FA}} := \{ \hat{P}_{\text{BLER}} \geq \tau_{\text{BLER}} \mid P_{\text{BLER}} < \tau_{\text{BLER}} \}$ and $\mathcal{E}_{\text{MD}} := \{ \hat{P}_{\text{BLER}} < \tau_{\text{BLER}} \mid P_{\text{BLER}} \geq \tau_{\text{BLER}} \}$. $\mathcal{E}_{\text{FA}}$ is specifically critical as it leads to changing the decision of the table-based algorithm. $\mathcal{E}_{\text{MD}}$ is also an important event, but it does not change the decision of the table-based algorithm.  
	\end{remark}
	
	\subsection{RI and CQI Selection}
	In 5G New Radio (NR) systems, the RI determines  the number of spatial data streams transmitted in parallel over the MIMO channel. The UE may strategically report a lower RI to enable a higher-order MCS on dominant layers, often yielding better aggregate throughput than a full-rank transmission bottlenecked by weak layers. \\
	We begin with a brief overview of the RI selection problem. The conventional table-based baseline algorithms typically first estimate the best CQI for each possible RI. Then, they select the RI that maximizes the estimated spectral efficiency (SE). Specifically, denoting the CQI selected by the table-based algorithm for $\mathrm{RI} = i$ as $\mathrm{CQI}^{(i)}$, the RI selected by the conventional table-based algorithm is given by 
	\begin{equation}
		r \triangleq \arg\max_{ i \in \{1, \cdots, r_{\text{max}}\}} \ \{ i \cdot \mathrm{SE}(i, \mathrm{CQI}^{(i)}) \},
	\end{equation}
	where $r_{\text{max}}$ denotes the maximum supported rank and $\mathrm{SE}(i, \mathrm{CQI}^{(i)})$ denotes the estimated SE per-layer when $\mathrm{RI} = i$ and $\mathrm{CQI}^{(i)}$ are selected. The per-layer SE is defined as $\mathrm{SE}(i, \mathrm{CQI}^{(i)}) = Q_m^{(i)} \cdot R^{(i)}$, where $Q_m^{(i)}$ and $R^{(i)}$ denote the modulation order and code rate, respectively, corresponding to the selected $\mathrm{CQI}^{(i)}$. Finally, the UE reports the  rank $\mathrm{RI} = r$ and the associated $\mathrm{CQI} = \mathrm{CQI}^{(r)}$ to the gNB.
	
	Our algorithm, \RICQI{}, extends \CQI{} to jointly select the RI and the CQI that maximize the estimated spectral efficiency while avoiding the excessively high BLER typically associated with conventional algorithms. Specifically, \RICQI{} refines the baseline CQI selected by the table-based algorithm, denoted as $\mathrm{CQI}^{(i)}$, to a more robust, possibly lower CQI denoted as $\mathrm{CQI}^{(i)}_{\text{L}}$. Then, we select the rank $r_{\text L}$ as follows
	\begin{equation}
		r_{\text L} \triangleq \arg\max\limits_{i \in \{1, \dots, r_{\text{max}}\}} \ \{ i \cdot \mathrm{SE}(i, \mathrm{CQI}^{(i)}_{\text L}) \}.
	\end{equation} 
	The UE then reports $r_{\text L}$ and $\mathrm{CQI}^{(r_{\text L})}_{\text L}$ to the gNB.
	\begin{algorithm}
		\small 
		\caption{RI-CQI Selection Algorithm: \RICQI{}}
		\label{alg:ri_cqi_tune}
		\begin{algorithmic}[1]
			\REQUIRE Current parameters $\bm{\theta}^{(t)}$,  table-based rank $r$ and table-based CQI for each rank $\{\text{CQI}^{(i)}\}_{i=1}^{r_{\text{max}}}$, feature vector $\bm{x}^{(t)}$, minimum CQI: $\text{CQI}_{\text{min}}$, BLER threshold $\tau_{\text{BLER}}$, maximum rank $r_{\text{max}}$, fine-tuning threshold $\delta$, learning rate $\alpha_f$ and BLER classes $\mathcal P$.
			\ENSURE Adjusted rank $r_{\text L}$, adjusted CQI denoted by $\text{CQI}^{(r_{\text L})}_{\text L}$, \\
			\hspace{\algorithmicindent} \ \ updated model parameters $\bm{\theta}^{(t+1)}$
			
			\STATE $r_{\text{low}} \gets \lceil r / 2 \rceil$, $r_{\text{high}} \gets \min(r_{\text{low}} + 2, r_{\text{max}})$ 
			\STATE $\text{SE}_{\text{max}} \gets -1$, $r_{\text L} \gets r, \quad \text{CQI}^{(r_{\text L})}_{\text L} \gets \text{CQI}^{(r)}$
			\FOR{$i = r_{\text{low}}$ \TO $r_{\text{high}}$}
			\STATE $\text{CQI}_{\text{test}} \gets \text{CQI}^{(i)}$
			
			\WHILE{\( \text{CQI}_{\text{test}} > \text{CQI}_{\text{min}} \)}
			\STATE Update rank to $i$ and CQI to $\text{CQI}_{\text{test}}$ in $\bm{x}^{(t)}$
			\IF{$ \BLER(\bm{\theta}^{(t)}, \bm{x}^{(t)},\mathcal P) < \tau_{\text{BLER}}$}
			\STATE \textbf{break} 
			\ENDIF
			\STATE $\text{CQI}_{\text{test}} \gets \text{CQI}_{\text{test}} - 1$ 
			\ENDWHILE
			
			\STATE $\text{SE} \gets i \cdot \text{SE}(i, \text{CQI}_{\text{test}})$
			\IF{$\text{SE} > \text{SE}_{\text{max}}$}
			\STATE $\text{SE}_{\text{max}} \gets \text{SE}$,  $r_{\text L} \gets i, \quad \text{CQI}^{(r_{\text L})}_{\text L} \gets \text{CQI}_{\text{test}}$
			\ENDIF
			\ENDFOR
			\STATE Transmit PDSCH utilizing final adjusted $r_{\text L}$ and $\text{CQI}^{(r_{\text L})}_{\text L}$
			\STATE Compute $n_{\text{ACK}}^{(t)}$ and $n_{\text{NACK}}^{(t)}$ to calculate empirical $P_{\text{BLER}}^{(t)}$
			\STATE Update rank to $r_{\text L}$ and CQI to $\text{CQI}^{(r_{\text L})}_{\text L}$ in $\bm{x}^{(t)}$ 
			\STATE $\hat{P}_{\text{BLER}}^{(t)}, \bm{\theta}^{(t+1)} \gets \name\left(\bm{\theta}^{(t)}, \bm{x}^{(t)}, P_{\text{BLER}}^{(t)}, \delta, \alpha_f \right)$
			
			\STATE \textbf{return} $r_{\text L}, \text{CQI}^{(r_{\text L})}_{\text L}, \bm{\theta}^{(t+1)}$
		\end{algorithmic}
	\end{algorithm}

	\RICQI{} yields superior throughput gains by jointly optimizing RI and CQI. Moreover, the search bounds $r_{\text{low}}$ and $r_{\text{high}}$ ensure the rank window covers at most 3 candidates, making \RICQI{} strictly three times the complexity of \CQI{} regardless of $r_{\text{max}}$.	
	\begin{remark}[\textbf{Motivation for \RICQI{}}]
		\RICQI{} is motivated by empirical observations of live networks, where the gNB typically adheres to the UE-reported RI but frequently overrides the reported CQI. This aligns with 3GPP specifications \cite{3gpp_ts38214}, which grant the gNB flexibility to dictate final transmission parameters. Through joint optimization, \RICQI{} generates a robust feedback pair that accommodates these practical behaviors, mitigating high-BLER events caused by stale table-based mappings. Crucially, while \CQI{} is restricted to a conservative back-off, \RICQI{} can dynamically increase the RI.
	\end{remark}

	\subsection{Complexity Analysis}
	We first discuss the potential benefits of the FF algorithm compared to the backpropagation algorithm as summarized in Table \ref{table:complexity_comparison}. Then, we analyze the complexity of the BLER prediction algorithm. Finally, we compare the complexity of the \CQI{} algorithm and the \RICQI{} algorithm as summarized in Table \ref{tab:complexity}.\\ 
	\textbf{Complexity of the FF and the BP algorithms}. A primary motivation for adopting the FF algorithm is its significant advantages for resource-constrained devices \cite{torres2025advancements,huang2025tinyfoa,de2023mu}. We discuss these aspects next.\\
	1) \textbf{Dynamic Memory (RAM) Footprint}:  BP requires storing all activations during the backward pass; hence, the required RAM scales with the network depth and layer sizes (i.e., $\sum_l M_l$ for an MLP). In contrast, \name{} trains each layer locally, which eliminates the necessity to store intermediate activations across the entire depth. Consequently, the peak memory footprint scales only with the size of the largest layer (i.e., $\max_l M_l$ for an MLP).\\
	2) \textbf{Implementation Simplicity}:  BP necessitates a complex automatic differentiation (Autodiff) engine to build a computational graph and execute the chain rule. This typically requires a full deep-learning framework (e.g., PyTorch, TensorFlow). \name{} avoids this overhead entirely by leveraging the FF algorithm which utilizes local sequential update rules that do not require global graph tracking.
	
	\begin{table}
		\centering
		\setlength{\tabcolsep}{2pt} 
		\renewcommand{\arraystretch}{0.9}
		\begin{tabular}{@{}lll@{}}
			\toprule
			\textbf{Feature} & \textbf{Standard Backpropagation} & \textbf{\name{} (FF-based)} \\
			\midrule
			Peak RAM     & $O\left(\sum_{l=1}^{L} M_l\right)$ & $O\left(\max_l M_l\right)$             \\
			Control Logic     & Complex Autodiff     &  Local Training    \\
			\bottomrule
		\end{tabular}
		\captionsetup{font=footnotesize}
		\caption{Complexity analysis of BP vs. FF for an MLP with $L$ layers, where $M_l$ is the width of the $l$-th layer.}
		\label{table:complexity_comparison}
	\end{table}
	
	\textbf{BLER Prediction Complexity.} 
	The computational cost of a single short-term BLER prediction is governed by the FF inference procedure over the quantized classes. We define $Q = \sum_{l=1}^L M_l M_{l-1}$ as the number of multiply-accumulate (MAC) operations required for a single forward pass through the MLP. Since \BLER{} evaluates $C$ discrete classes, the total inference complexity is exactly $C \cdot Q$ MAC operations. \\
	\textbf{Complexity of \CQI{}.} 
	The complexity of the \CQI{} back-off algorithm (Alg.~\ref{alg:cqi_backoff}) is divided across the CQI selection and the post-transmission fine-tuning. The maximum number of pre-transmission search iterations is bounded by $K := \text{CQI}^{(r)} - \text{CQI}_{\text{min}}$. By setting $\text{CQI}_{\text{min}} = \text{CQI}^{(r)} - 1$, the worst-case pre-transmission overhead is strictly limited to exactly one BLER prediction (i.e., $C \cdot Q$ MACs). Combined with the two post-transmission forward passes (one for the positive sample and one for the negative sample) required for fine-tuning, the worst-case total complexity per CSI-RS period is exactly $(C + 2)Q$ MAC operations. \\
	\textbf{Complexity of \RICQI{}.} 
	The joint \RICQI{} selection framework (Alg.~\ref{alg:ri_cqi_tune}) introduces an outer search over a localized rank window $[r_{\text{low}}, r_{\text{high}}]$ covering at most 3 candidate ranks. Since each candidate rank invokes at most one pre-transmission BLER prediction (given $K=1$), the joint search performs a maximum of 3 predictions. \\
	\textbf{Fine-tuning Complexity.}
	Let $N_{\text{total}} = \sum_{l=1}^L M_l(M_{l-1} + 1)$ denote the total number of model parameters, including biases. The gradient computation in \eqref{eqn:gradients} requires two outer products. Since each outer product generates a matrix of size $M_l \times (M_{l-1}+1)$, the computation requires $2N_{\text{total}}$ operations. Including the summation of these gradients and the parameter update, the total fine-tuning cost is approximately $4N_{\text{total}}$ operations. Crucially, when utilizing the one-step sign-update variant in \eqref{eqn:sign_update}, the final parameter update requires 0 additional multiplications and $N_{\text{total}}$ additions (or bit-sign extractions). 
	
	\begin{table}
		\centering
		\setlength{\tabcolsep}{3pt} 
		\renewcommand{\arraystretch}{1.3}
		\resizebox{\columnwidth}{!}{%
			\begin{tabular}{@{}lcccc@{}}
				\toprule
				\multirow{2}{*}{\textbf{Algorithm}} & \multicolumn{2}{c}{\textbf{FF Inferences}} & \textbf{Update Cost} & \multirow{2}{*}{\textbf{Total MACs}} \\
				\cmidrule(lr){2-3} \cmidrule(lr){4-4}
				& \shortstack{\textbf{Inference}} & \shortstack{\textbf{Fine-tuning}} & \shortstack{\textbf{One-step Update}} & \\
				\midrule
				\BLER{}              & $C \cdot Q$ & 0 & 0 & $C \cdot Q$ \\
				\CQI{}               & $C \cdot Q$ & $2Q$ & $4N_{\text{total}}$ & $(C+2)Q + 4N_{\text{total}}$ \\
				\RICQI{}             & $3C \cdot Q$ & $2Q$ & $4N_{\text{total}}$ & $(3C+2)Q + 4N_{\text{total}}$ \\
				\bottomrule
			\end{tabular}
		}
		\captionsetup{font=footnotesize}
		\caption{Worst-case complexity analysis where $Q = \sum_{l=1}^L M_l M_{l-1}$ and $N_{\text{total}} = \sum_{l=1}^L M_l(M_{l-1} + 1)$.}
		\label{tab:complexity}
	\end{table}
	Finally, as discussed in Remark \ref{remark:FA_MD}, the number of quantized BLER classes could potentially be restricted to $C = 2$, further minimizing the computational overhead of the FF inference.
	
	\section{Simulation Results}
	\label{sec:results}
	We now evaluate the performance of the proposed fine-tuning algorithm for BLER prediction and link adaptation. 
	\subsection{Training and Test Settings}
	We train and test our model on tapped delay line (TDL) channels~\cite{3gppTR38901}, as summarized in Table~\ref{table:channel_settings}. In TDL channels, the letter (A, B, C) indicates the channel profile, while the number (10, 30, 50, ...) specifies the root mean square delay spread (RMS DS) in nanoseconds. We consider various 3GPP antenna correlation scenarios \cite{3gppTR38901}.  
	For both the training and the test settings, the bandwidth is $100$ MHz, the sub-carrier spacing is $30$ KHz and the maximum supported RI $r_{\text{max}}$ is $4$. Our BLER threshold is set as $\tau_{\text{BLER}} = 0.9$ and the minimum CQI of \CQI{} is set as $\text{CQI}_{\text{min}} =\max( \text{CQI}^{(r)} - 1, 1)$, where $\text{CQI}^{(r)}$ is the table-based selection.
	
	\begin{table}
		\centering
		\setlength{\tabcolsep}{2pt} 
		\begin{tabular}{@{}p{2.3cm}p{1.8cm}p{4cm}@{}} 
			\toprule
			\textbf{Parameter} & \textbf{Training} & \textbf{Test} \\
			\midrule
			Channels & TDL-A30 & TDL-A10, TDL-A30, TDL-B50, TDL-B100, TDL-B200, TDL-C200 \\
			SNR & Low (0–12 dB) & Low/Medium/High (0–40 dB) \\
			Delay Profile & Low Delay & Low/High Delay \\
			Doppler Frequency & Low (10 Hz) & Low/Medium (10–50 Hz) \\
			Antenna Correlation & Low & Low/Medium/High \\
			CSI-RS Period & 80 ms & 10, 40 or 80 ms\\
			\bottomrule
		\end{tabular}
		\captionsetup{font=footnotesize}
		\caption{Training and testing configurations. }
		\label{table:channel_settings}
	\end{table}
	
	\subsection{BLER Prediction}
	\label{subsec:bler_experiment}
	Our goal is to simulate scenarios in which the ML model is trained in an environment but evaluated in a different one to mimic the training-test mismatch problem. We use $12$ features for the BLER prediction ML model \cite{Ali2026lighttune}, hence the input size is $13$ since a label is attached to the feature vector. \\
	\textbf{Features}. We utilize features that capture both channel conditions and transmission parameters that include the CSI-RS SNR and the CSI-RS capacity, which helps characterize antenna correlation levels. We also incorporate the delay spread and Doppler frequency estimates to distinguish between different channel profiles. Additionally, we use the instantaneous PDSCH SNR along with the PDSCH SNR values from the three most recent transmissions, which aid in predicting the likelihood of NACKs. This is completed by the current $\text{RI}$ and $\text{CQI}$, the number of allocated resource blocks (RBs), and the number of demodulation reference signal (DMRS) symbols.
	
	\begin{table}
		\centering
		\setlength{\tabcolsep}{3pt} 
		\renewcommand{\arraystretch}{0.85} 
		\begin{tabular}{@{}ll@{}} 
			\toprule
			\textbf{Parameter} & \textbf{Value} \\
			\midrule
			Neural Network Size & \texttt{[13, 32, 32]} \\
			Offline Learning Rate $\alpha$ & $0.03$ \\
			Online Learning Rate $\alpha_{f}$ & $0.03$ \\
			Fine-tuning Threshold $\delta$ & $0.3$ \\
			Training Optimizer & Adam ($\beta_1 = 0.9, \beta_2 = 0.999, \epsilon = 10^{-8}$)\\
			Threshold $T$ & $9$ \\
			Epochs & $22{,}000$ \\
			Training Samples & $83{,}200$ \\
			BLER Classes & $\mathcal{P} = \{0, 0.1, 0.2, \cdots, 0.9\}$ \\
			\bottomrule
		\end{tabular}
		\captionsetup{font=footnotesize}
		\caption{Hyperparameters for the BLER prediction algorithm}
		\label{table:finetuning_params}
	\end{table}
	
	\begin{figure}
		\centering
		\captionsetup{font=footnotesize}
		\begin{subfigure}[b]{0.49\linewidth}
			\includegraphics[width=\linewidth]{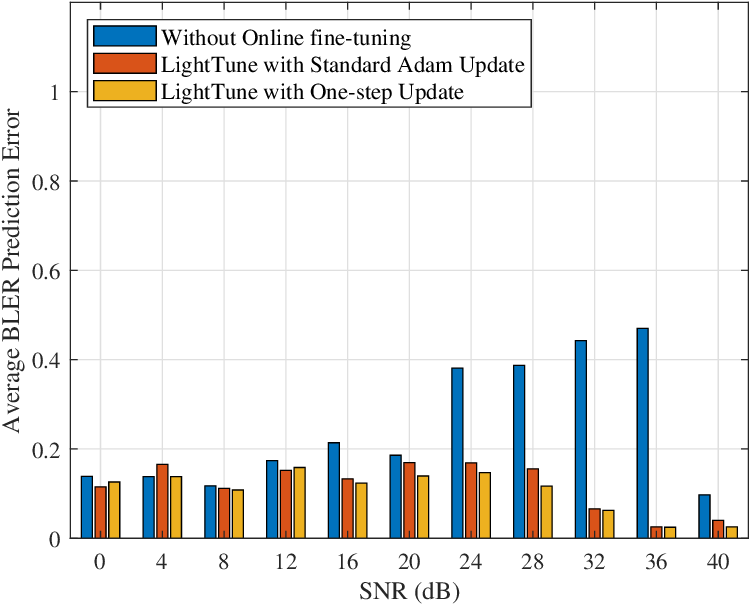}
			\caption{TDL-A30 \\ (Doppler freq. = 10 Hz)}
			\label{fig:bler_bar_tdla}
		\end{subfigure}
		\begin{subfigure}[b]{0.49\linewidth}
			\includegraphics[width=\linewidth]{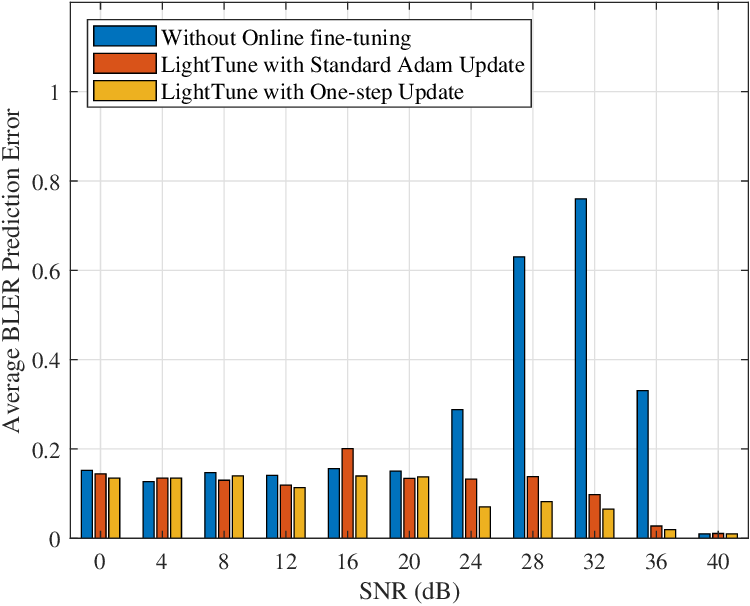}
			\caption{TDL-B50 \\ (Doppler freq. = 30 Hz)}
			\label{fig:bler_bar_tdlb50}
		\end{subfigure}
		\captionsetup{font=footnotesize}
		\caption{BLER prediction error with and without online fine-tuning, uniform sampling, under low antenna correlation, CSI-RS period = 80 ms.}
		\label{fig:bler_bar_combined}
	\end{figure}
	
	\noindent \textbf{BLER Prediction with and without fine-tuning}. Our results show that the one-step update approach is the best in terms of the BLER prediction error as shown in Fig. \ref{fig:bler_bar_combined} and the FA probability as shown in Fig. \ref{fig:FA_MD_comparison}. Fig. \ref{fig:bler_bar_tdla} shows that for SNRs between $0$ and $12$ dB, the average BLER prediction error remains largely unchanged with or without online fine-tuning, as the baseline model was trained on TDL-A30 data within this exact range. However, at SNRs above $12$ dB, enabling online fine-tuning yields a significant decrease in prediction error. Specifically, the average BLER prediction error across all SNRs is reduced by $43.5\%$ when utilizing the one-step update variant. Fig. \ref{fig:bler_bar_tdlb50} illustrates a substantial decrease in the BLER prediction error when online fine-tuning is enabled. Because the TDL-B50 profile was not included in the offline training phase, the resulting training-test mismatch severely degrades the  offline model's accuracy. With \name, the average BLER prediction error is decreased by $36\%$.\\
	Since our BLER prediction results show that the one-step variant of \name{} leads to a lower BLER prediction error in most cases, we only consider this variant for CQI selection and for joint RI and CQI selection. 
	
	\begin{figure}
		\captionsetup{font=footnotesize}
		\centering
		\begin{subfigure}{0.49\linewidth}
			\centering
			\includegraphics[width=\linewidth]{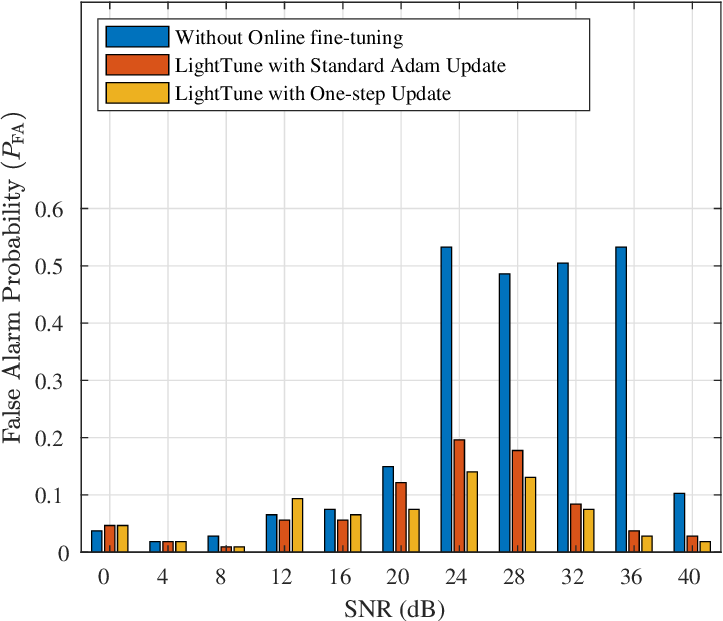}
			\captionsetup{font=footnotesize}
			
			\caption{TDL-A30 (Dop. freq. = 10 Hz).}
			\label{fig:tdla30_fa}
		\end{subfigure}
		\hfill
		\begin{subfigure}{0.49\linewidth}
			\centering
			\includegraphics[width=\linewidth]{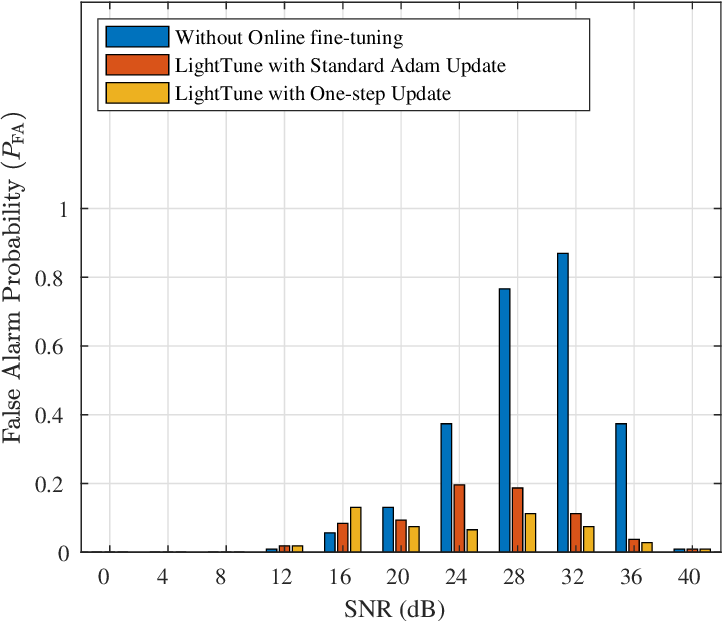}
			\captionsetup{font=footnotesize}
			
			\caption{TDL-B50 (Dop. freq. = 30 Hz)}
			\label{fig:tdlb50_fa}
		\end{subfigure}
		\captionsetup{font=footnotesize}
		\caption{False alarm (FA) probability with and without online fine-tuning under low antenna correlation, CSI-RS period = 80 ms, where $\tau_{\text{BLER}}=0.9$.}
		\label{fig:FA_MD_comparison}
	\end{figure}
	
	\begin{figure}
		\captionsetup{font=footnotesize}
		\centering
		\begin{subfigure}{0.49\linewidth}
			\centering
			\includegraphics[width=\linewidth]{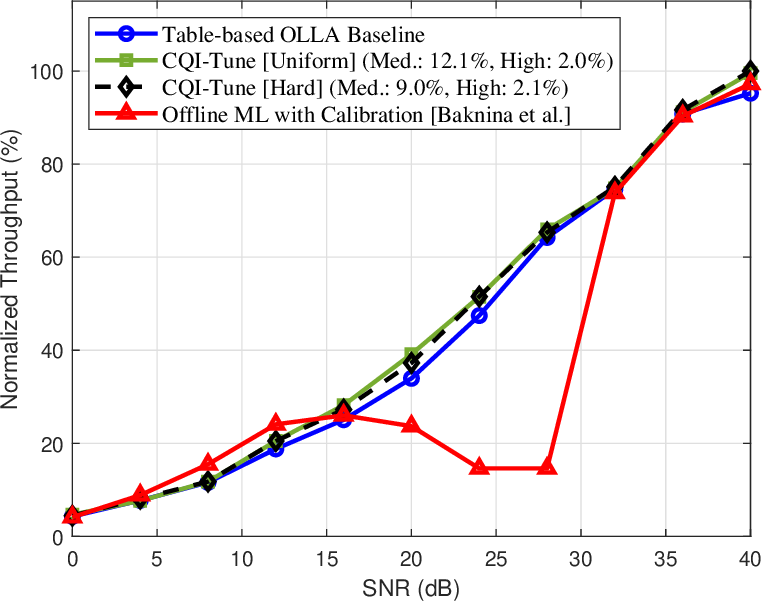}
			\captionsetup{font=footnotesize}
			
			\caption{TDL-B50 (Dop. freq. = 30 Hz).}
			\label{fig:tdlb_throughput}
		\end{subfigure}
		\hfill
		\begin{subfigure}{0.49\linewidth}
			\centering
			\includegraphics[width=\linewidth]{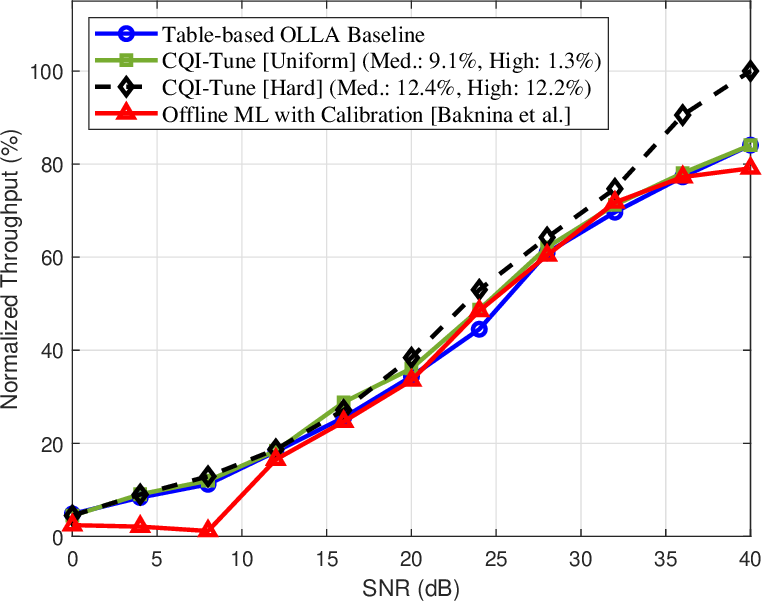}
			\captionsetup{font=footnotesize}
			
			\caption{TDL-C200 (Dop. freq. = 50 Hz)}
			\label{fig:tdlc_throughput}
		\end{subfigure}
		\captionsetup{font=footnotesize}
		\caption{Throughput  of \CQI{}, under low antenna  correlation with CSI-RS period = 80 ms.}
		\label{fig:throughput_comparison}
	\end{figure}
	
	\subsection{CQI Selection}
	We first show the effect of \name{} on the FA probability in Fig. \ref{fig:FA_MD_comparison} under channel and SNR mismatch. As  discussed in Remark \ref{remark:FA_MD}, the FA probability is crucial in CQI selection. Since the ML BLER prediction model has been trained offline within the low SNR range, we see that the FA probability is low even without fine-tuning in the low SNR range. However, fine-tuning is essential in the other SNR ranges to keep the FA probability as low as possible. We next compare the throughput performance of three approaches with $\text{CQI}$ reporting based on: a table-based OLLA method similar to the algorithm of~\cite{peralta2022outer}, the offline ML-based method~\cite{baknina2020adaptive,ali2026online} with calibration/adaptation, and the OLLA method augmented with the proposed backoff mechanism described in Algorithm~\ref{alg:cqi_backoff}. We provide a comparison in Fig. \ref{fig:throughput_comparison}. We show the throughput gains  in the medium SNR range (16 to 24 dB) and the high SNR range (28 to 40 dB). 
	
	\textbf{CQI Selection Performance.} As detailed in Table \ref{table:cqi_sampling_gains}, \CQI{} consistently improves throughput across both channel profiles. For TDL-B50, medium SNR gains reach up to $12.1\%$ (uniform sampling), while high SNR gains plateau around $2.0\%$. Under the TDL-C200 profile, hard sampling significantly outperforms the uniform scheme in the high SNR regime, achieving a gain of $12.2\%$.
	
	\begin{table}[htb!]
		\centering
		\setlength{\tabcolsep}{6pt}
		\renewcommand{\arraystretch}{1.1}
		\begin{tabular}{@{}lcccc@{}}
			\toprule
			& \multicolumn{2}{c}{\textbf{Medium SNR Gain}} & \multicolumn{2}{c}{\textbf{High SNR Gain}} \\
			\cmidrule(lr){2-3} \cmidrule(lr){4-5}
			\textbf{Channel} & \textbf{Uniform} & \textbf{Hard} & \textbf{Uniform} & \textbf{Hard} \\
			\midrule
			TDL-B50, 30 Hz & 12.1\% & 9.0\% & 2.0\% & 2.1\% \\
			TDL-C200, 50 Hz & 9.1\% & 12.4\% & 1.3\% & 12.2\% \\
			\bottomrule
		\end{tabular}
		\captionsetup{font=footnotesize}
		\caption{Throughput gains of \CQI{} with uniform and hard  sampling over the table-based OLLA, low correlation, CSI-RS period = 80 ms.}
		\label{table:cqi_sampling_gains}
	\end{table}
	The throughput gains for both sampling strategies across the medium and high SNR regimes are summarized in Table~\ref{table:cqi_sampling_gains}. We also note that the offline ML method \cite{baknina2020adaptive} shows degraded performance due to a training-test mismatch, having been trained exclusively on TDL-A30 low SNR data. This occurs despite using ACKs/NACKs to calibrate the model output, highlighting the need for an online learning mechanism to help offline ML models generalize to the unseen conditions.

\textbf{RI and CQI Selection Performance}. Figs.~\ref{fig:low_corr_throughput}--\ref{fig:throughput_10ms_CSI_RS} and Table~\ref{table:consolidated_throughput_gains} compare \RICQI{} against \CQI{} with uniform negative sampling across varying CSI-RS periods, channel profiles, and antenna correlations. \RICQI{} achieves throughput gains of up to $15.5\%$ over the OLLA baseline. The principal advantage of \RICQI{} over \CQI{} is concentrated in the high-SNR regime, where the joint RI--CQI search can \emph{increase} the reported rank of the baseline algorithm. At medium SNR, the rank selected by the table-based algorithm is already appropriate and the CQI back-off alone is sufficient; consequently, \CQI{} is competitive with or occasionally exceeds \RICQI{} in that regime.
			\begin{figure*}
		\captionsetup{font=footnotesize}
		\centering
		\begin{subfigure}{0.24\linewidth}
			\centering
			\includegraphics[width=\linewidth]{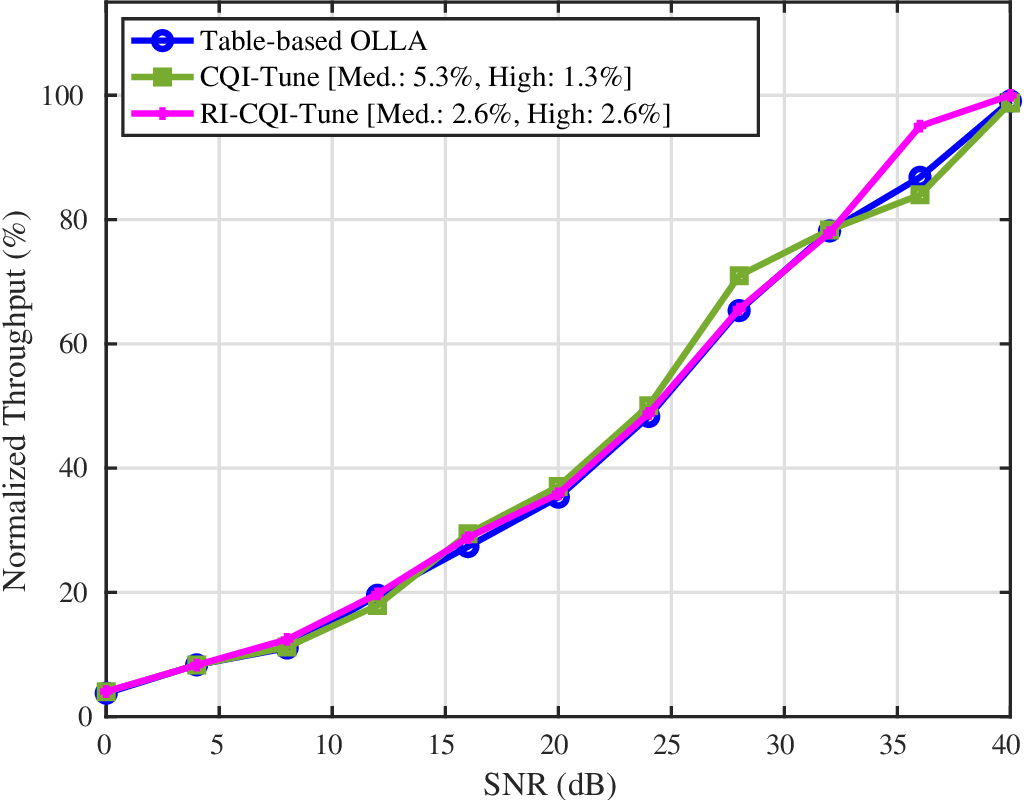}
			\caption{TDL-A10 \\ (Dop. freq. = 20 Hz).}
			\label{fig:tdla50_throughput_80ms_CSIRS_106RBs}
		\end{subfigure}
		\hfill
		\begin{subfigure}{0.24\linewidth}
			\centering
			\includegraphics[width=\linewidth]{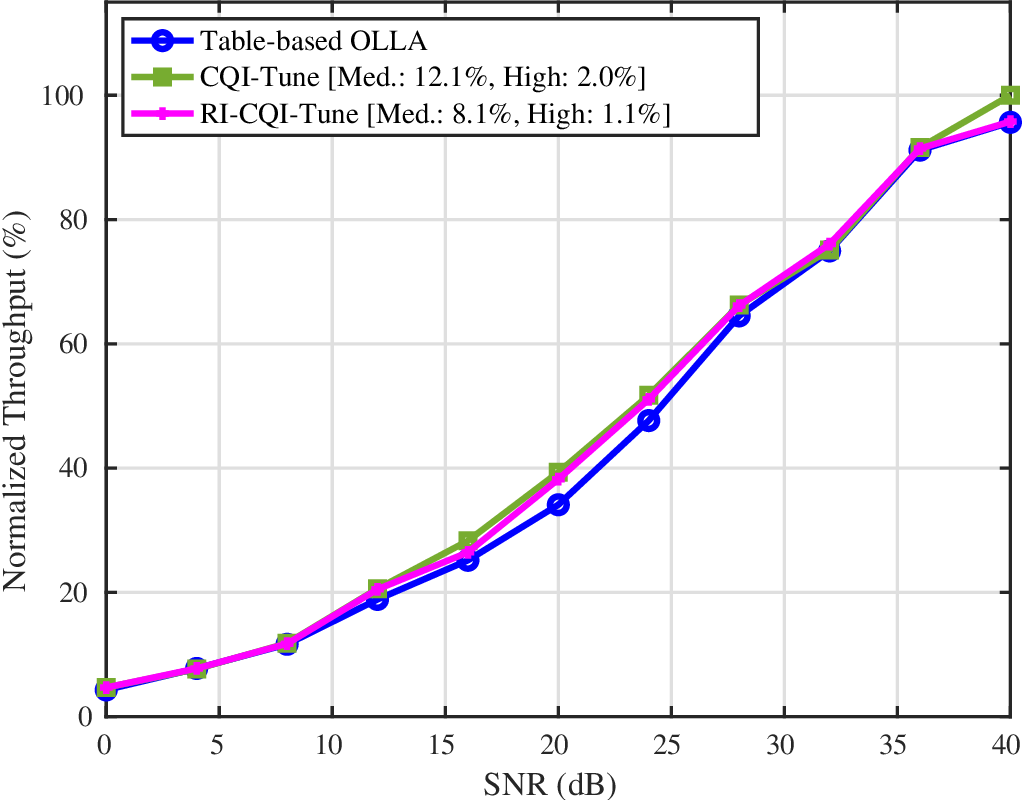}
			\caption{TDL-B50 \\ (Dop. freq. = 30 Hz).}
			\label{fig:tdlb50_throughput_80ms_CSIRS_106RBs}
		\end{subfigure}
		\hfill
		\begin{subfigure}{0.24\linewidth}
			\centering
			\includegraphics[width=\linewidth]{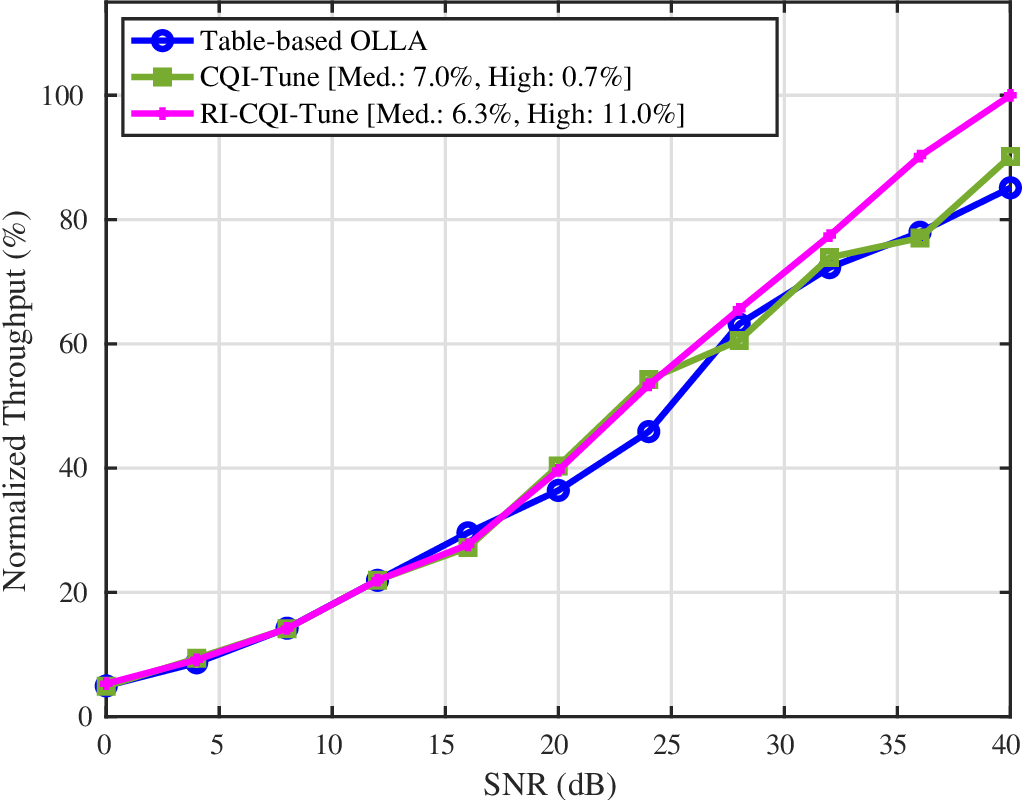}
			\caption{TDL-B200 \\ (Dop. freq. = 50 Hz).}
			\label{fig:tdlb200_throughput_80ms_CSIRS_106RBs}
		\end{subfigure}
		\hfill
		\begin{subfigure}{0.24\linewidth}
			\centering
			\includegraphics[width=\linewidth]{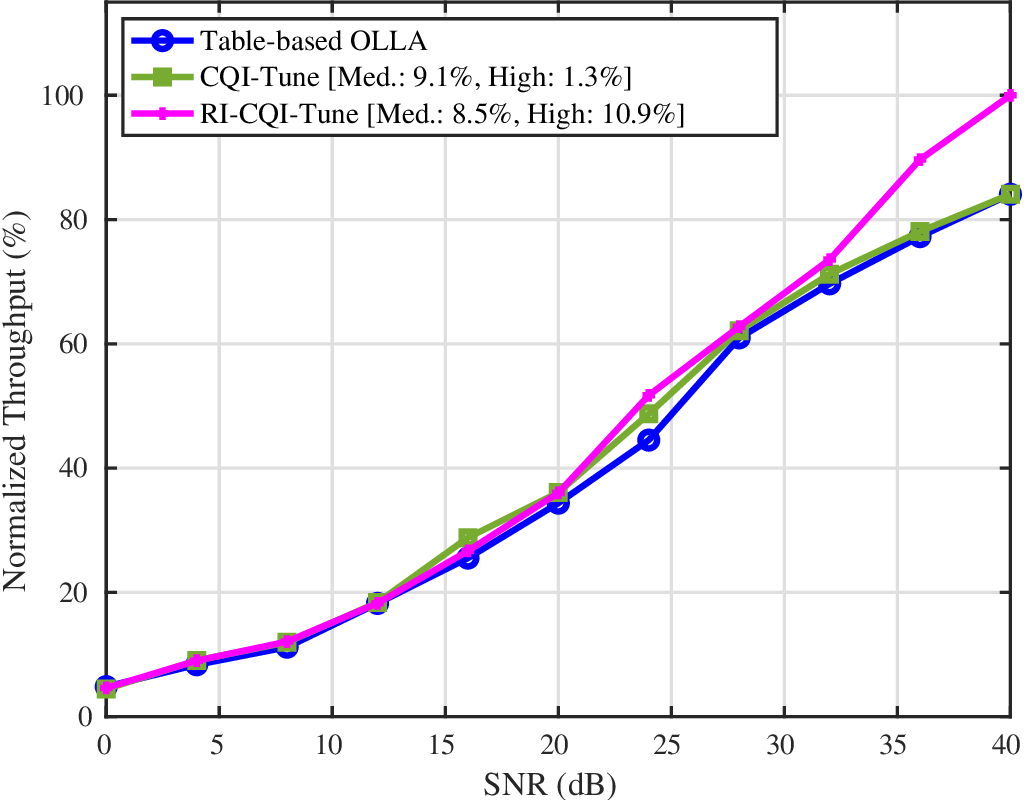}
			\caption{TDL-C200 \\ (Dop. freq. = 50 Hz).}
			\label{fig:tdlc200_throughput_80ms_CSIRS_106RBs}
		\end{subfigure}
		\captionsetup{font=footnotesize}
		\caption{Throughput of \RICQI{} and \CQI{} with uniform sampling under low antenna correlation with CSI-RS period = 80 ms.}
		\label{fig:low_corr_throughput}
	\end{figure*}
			\begin{figure}
		\captionsetup{font=footnotesize}
		\centering
		\begin{subfigure}{0.4935\linewidth}
			\centering
			\includegraphics[width=\linewidth]{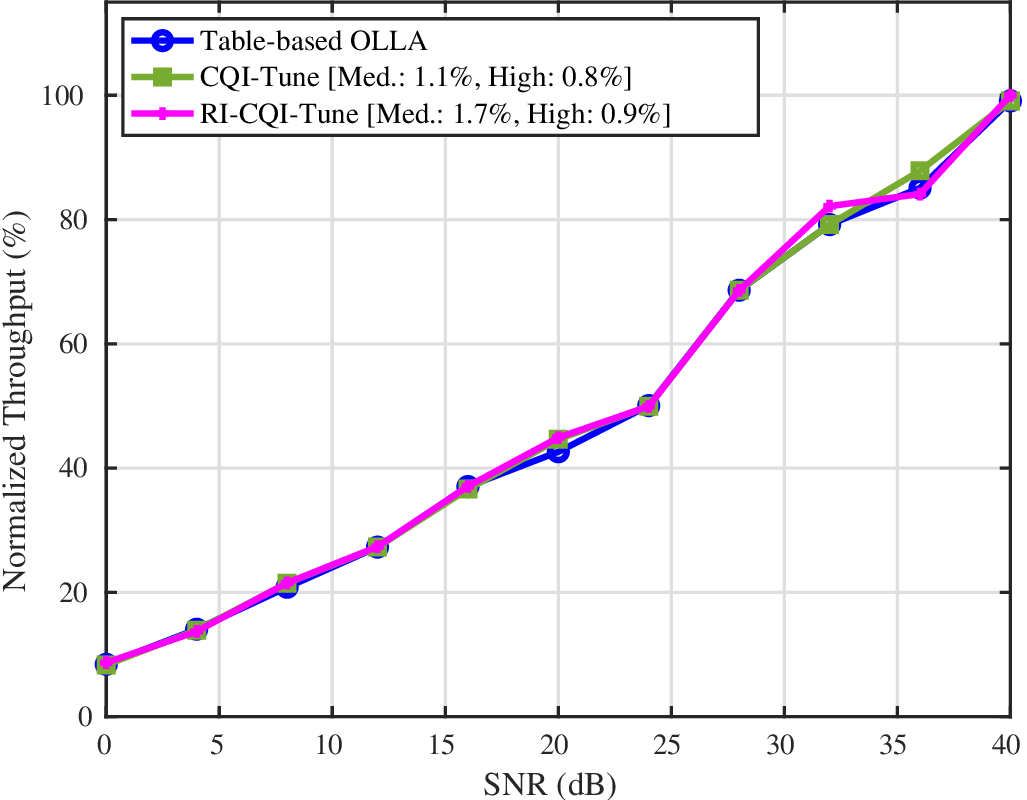}
			\caption{TDL-B50 \\ (Dop. freq. = 30 Hz).}
			\label{fig:tdlb50_throughput_80ms_CSIRS_106RBs_med_corr}
		\end{subfigure}
		\hfill
		\begin{subfigure}{0.4935\linewidth}
			\centering
			\includegraphics[width=\linewidth]{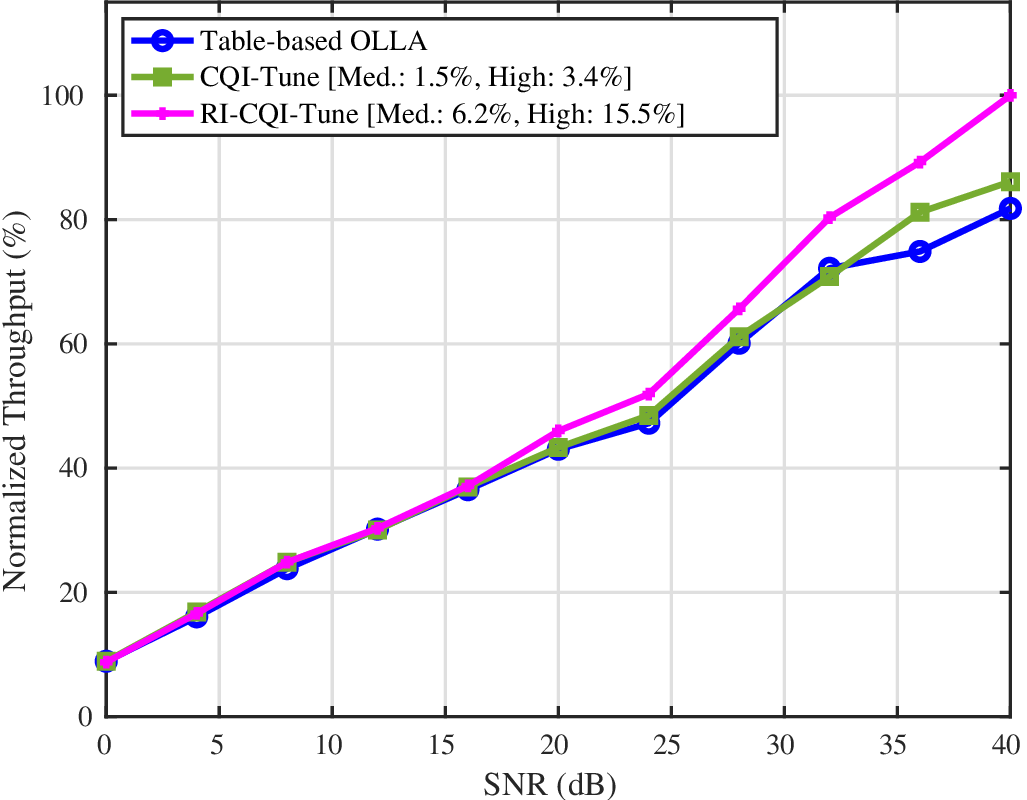}
			\caption{TDL-B200 \\ (Dop. freq. = 50 Hz).}
			\label{fig:tdlb200_throughput_80ms_CSIRS_106RBs_med_corr}
		\end{subfigure}
		\captionsetup{font=footnotesize}
		\caption{Throughput of \RICQI{} and \CQI{} with uniform sampling under medium antenna correlation with CSI-RS period = 80 ms.}
		\label{fig:med_corr_throughput}
	\end{figure}
	
	\begin{figure}
		\captionsetup{font=footnotesize}
		\centering
		\begin{subfigure}{0.4935\linewidth}
			\centering
			\includegraphics[width=\linewidth]{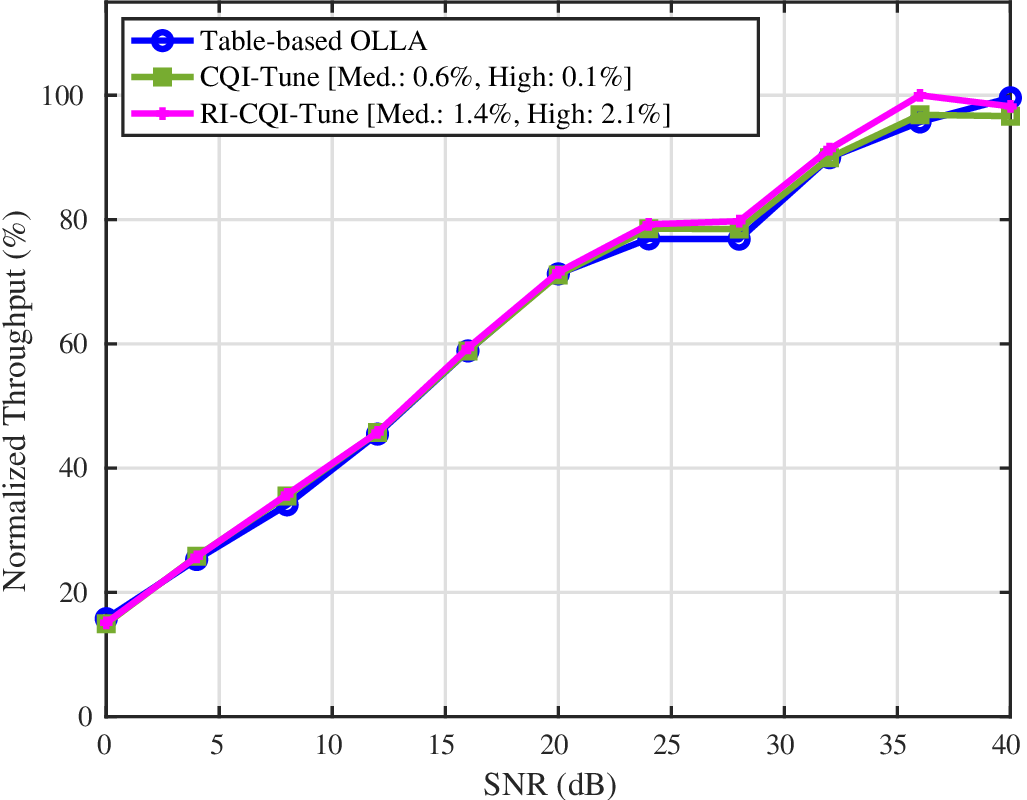}
			\caption{TDL-B50 \\ (Dop. freq. = 30 Hz).}
			\label{fig:tdlb50_throughput_80ms_CSIRS_106RBs_high_corr}
		\end{subfigure}
		\hfill
		\begin{subfigure}{0.4935\linewidth}
			\centering
			\includegraphics[width=\linewidth]{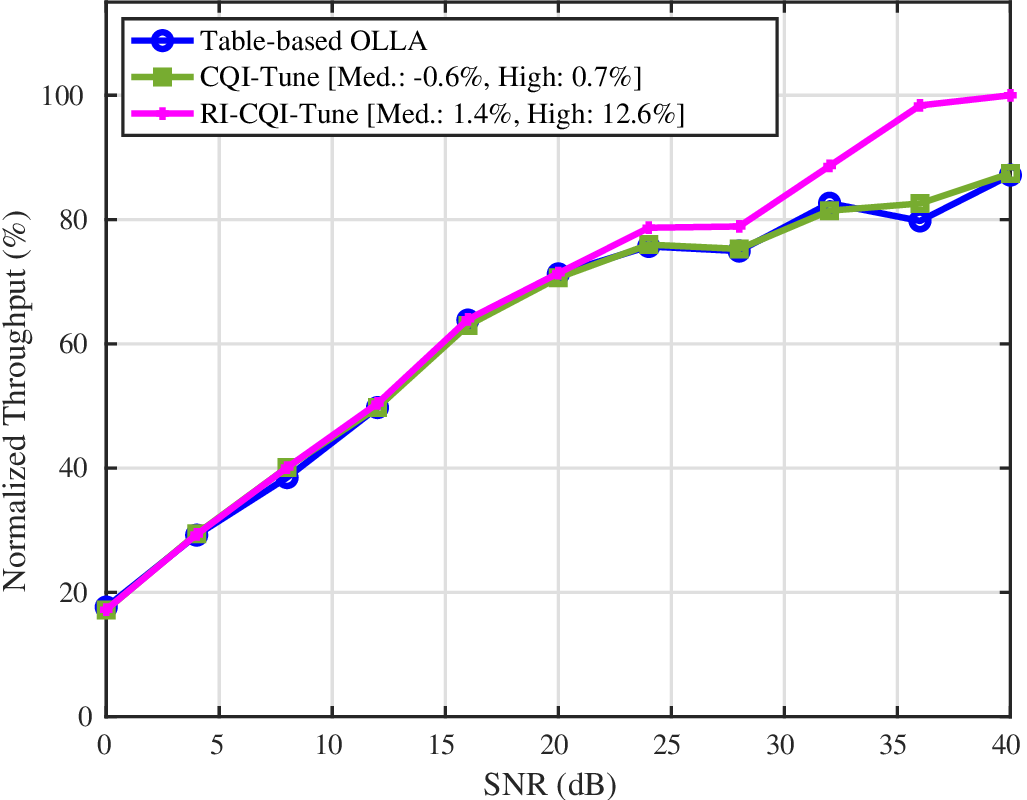}
			\caption{TDL-B200 \\ (Dop. freq. = 50 Hz).}
			\label{fig:tdlb200_throughput_80ms_CSIRS_106RBs_high_corr}
		\end{subfigure}
		\captionsetup{font=footnotesize}
		\caption{Throughput of \RICQI{} and \CQI{} with uniform sampling under high antenna correlation with CSI-RS period = 80 ms.}
		\label{fig:high_corr_throughput}
	\end{figure}
		\begin{figure*}
		\centering
		\captionsetup{font=footnotesize, justification=centering} 
		\begin{subfigure}{0.24\linewidth}
			\centering
			\includegraphics[width=\linewidth]{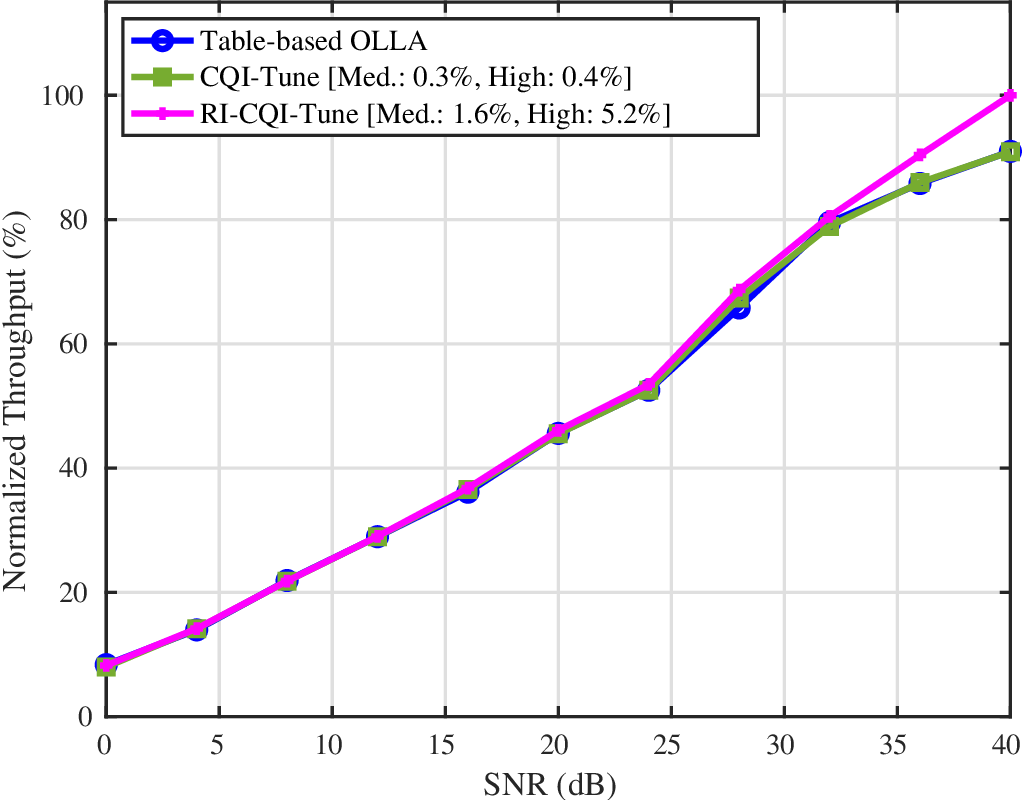}
			\caption{TDL-B50, Medium correlation \\ (Dop. freq. = 30 Hz).}
			\label{fig:tdlb50_throughput_40ms_CSIRS_106RBs_med_corr}
		\end{subfigure}
		\hfill
		\begin{subfigure}{0.24\linewidth}
			\centering
			\includegraphics[width=\linewidth]{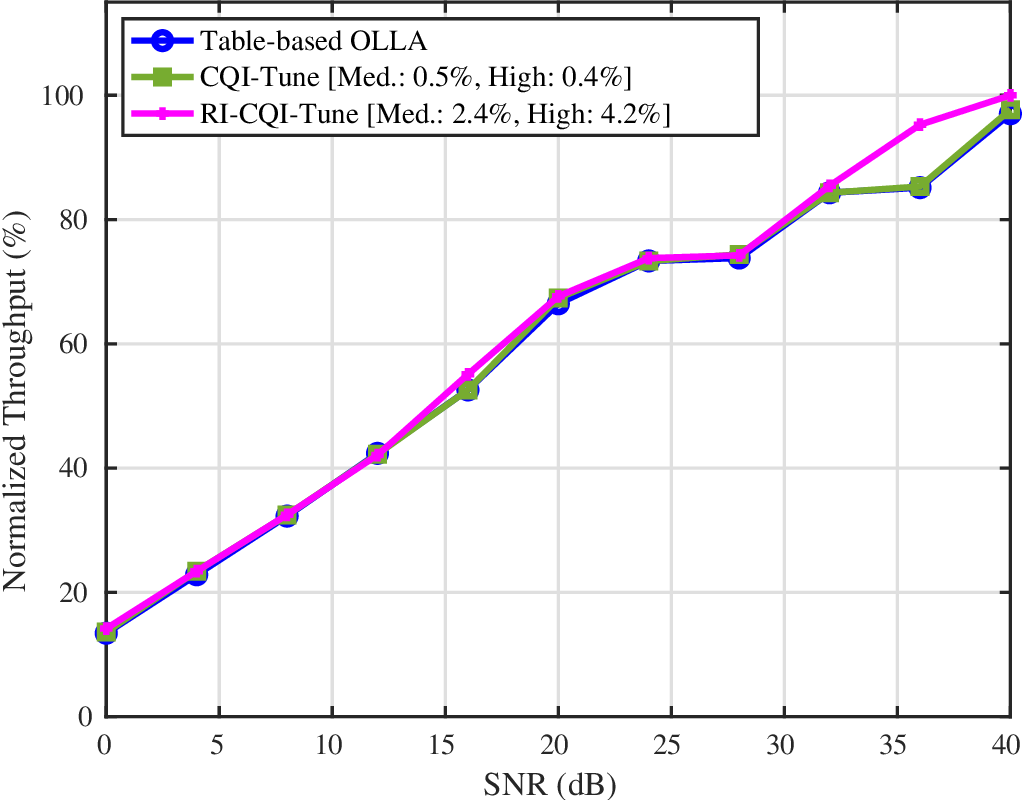}
			\caption{TDL-B50, High correlation \\ (Dop. freq. = 30 Hz).}
			\label{fig:tdlb50_throughput_40ms_CSIRS_106RBs_high_corr}
		\end{subfigure}
		\hfill
		\begin{subfigure}{0.24\linewidth}
			\centering
			\includegraphics[width=\linewidth]{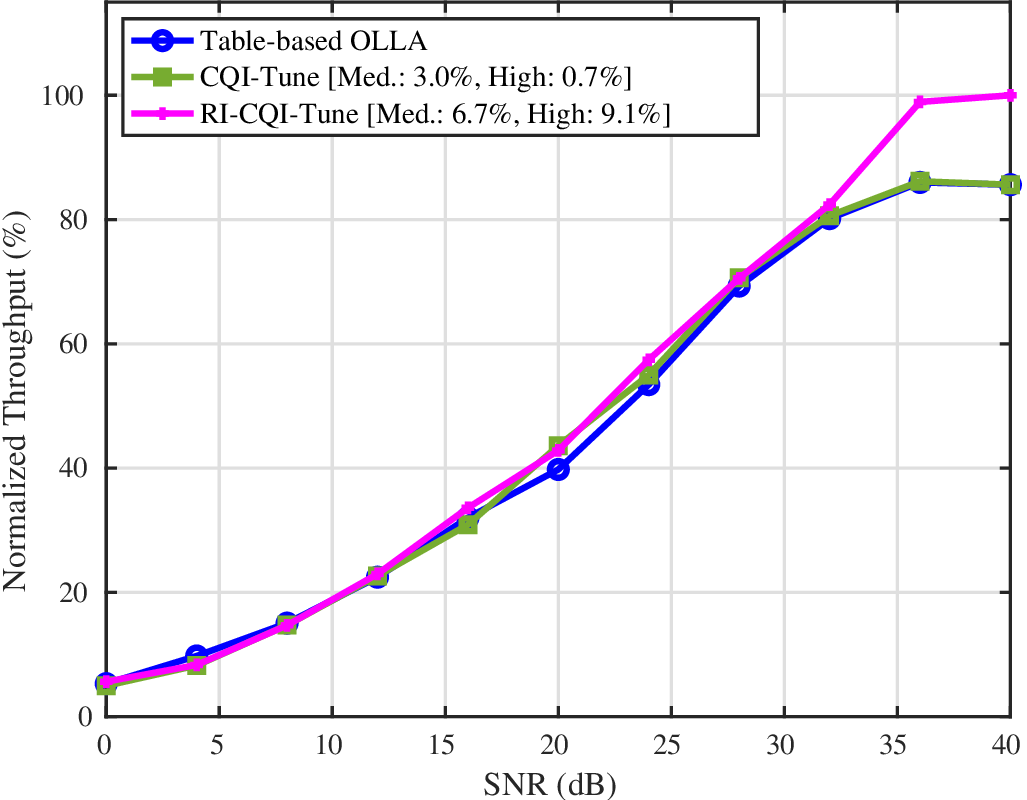}
			\caption{TDL-C200, Low correlation \\ (Dop. freq. = 50 Hz).}
			\label{fig:tdlc200_throughput_40ms_CSIRS_106RBs_low_corr}
		\end{subfigure}
		\hfill
		\begin{subfigure}{0.24\linewidth}
			\centering
			\includegraphics[width=\linewidth]{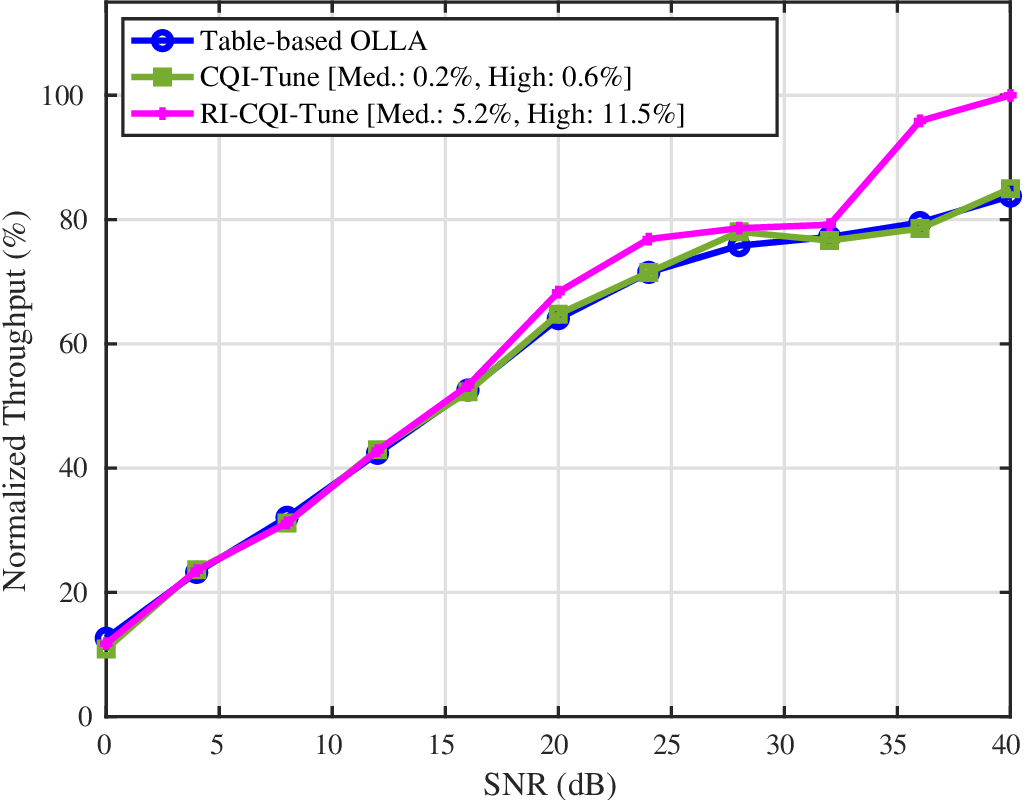}
			\caption{TDL-C200, High correlation \\ (Dop. freq. = 50 Hz).}
			\label{fig:tdlc200_throughput_40ms_CSIRS_106RBs_high_corr}
		\end{subfigure}
		\captionsetup{font=footnotesize}
		\caption{Throughput of \RICQI{} and \CQI{} with uniform sampling with CSI-RS period = 40 ms.}
		\label{fig:throughput_40ms_CSI_RS}
	\end{figure*}
	
	\begin{figure*}
		\captionsetup{font=footnotesize, justification=centering}
		\centering
		\begin{subfigure}{0.24\linewidth}
			\centering
			\includegraphics[width=\linewidth]{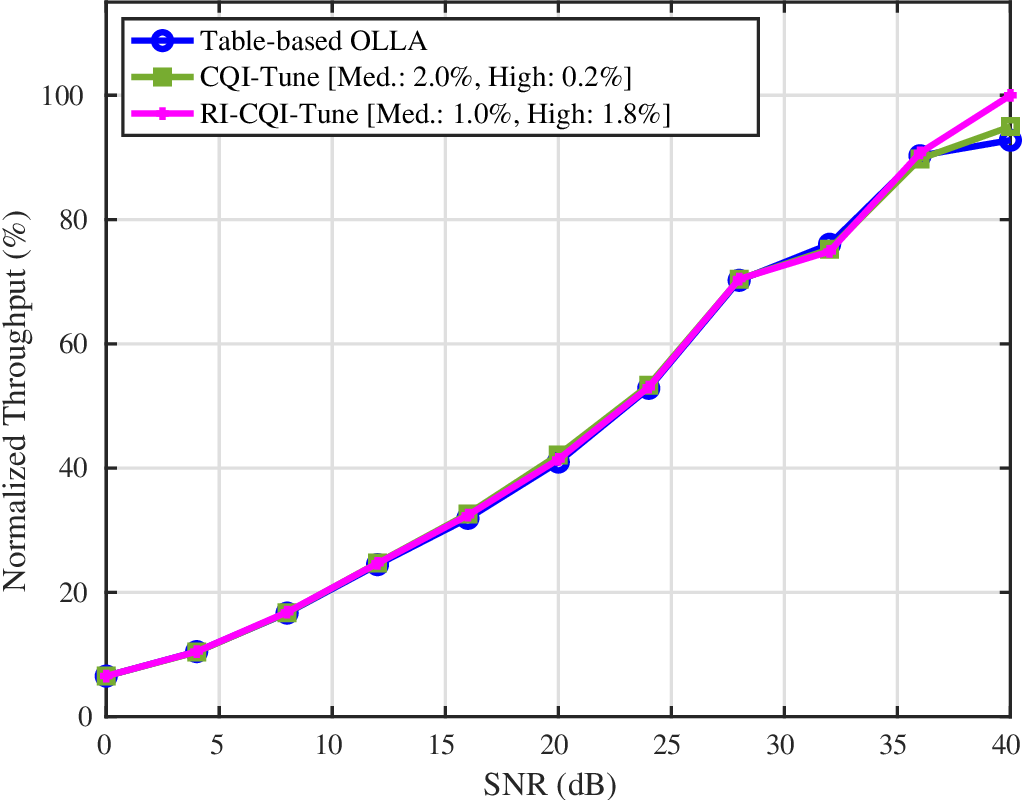}
			\caption{TDL-B50, Low correlation \\ (Dop. freq. = 30 Hz).}
			\label{fig:tdlb50_throughput_10ms_CSIRS_106RBs_low_corr}
		\end{subfigure}
		\hfill
		\begin{subfigure}{0.24\linewidth}
			\centering
			\includegraphics[width=\linewidth]{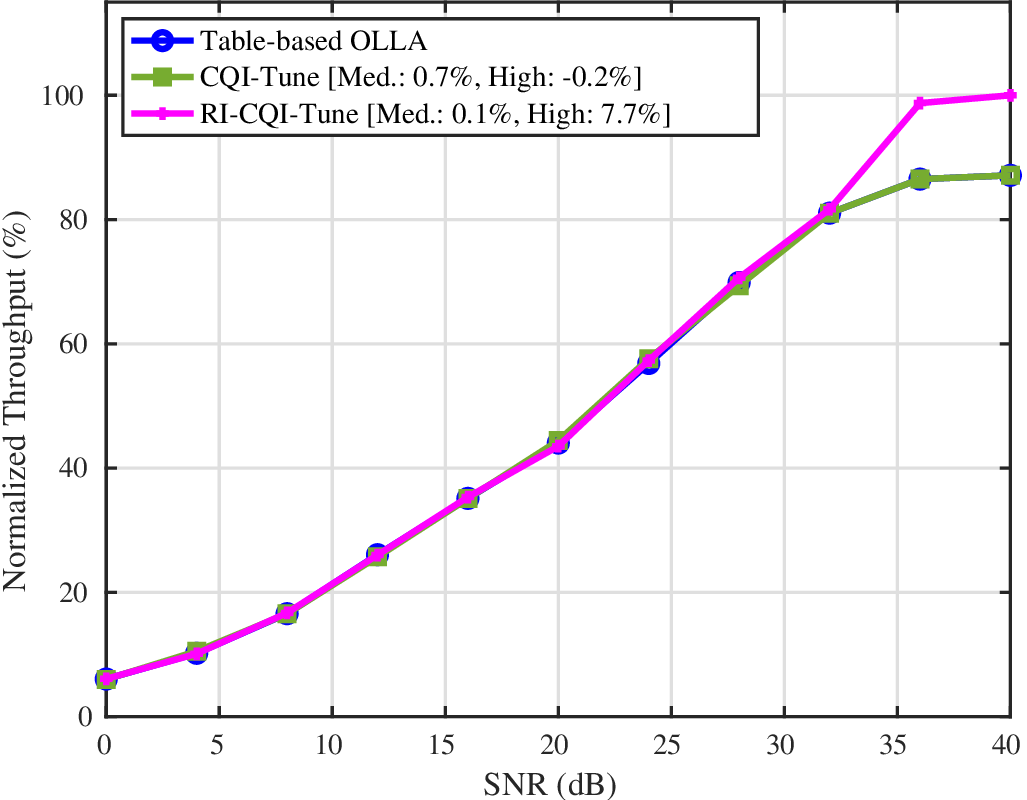}
			\caption{TDL-C200, Low correlation \\ (Dop. freq. = 50 Hz).}
			\label{fig:tdlc200_throughput_10ms_CSIRS_106RBs_low_corr}
		\end{subfigure}
		\hfill 
		\begin{subfigure}{0.24\linewidth}
			\centering
			\includegraphics[width=\linewidth]{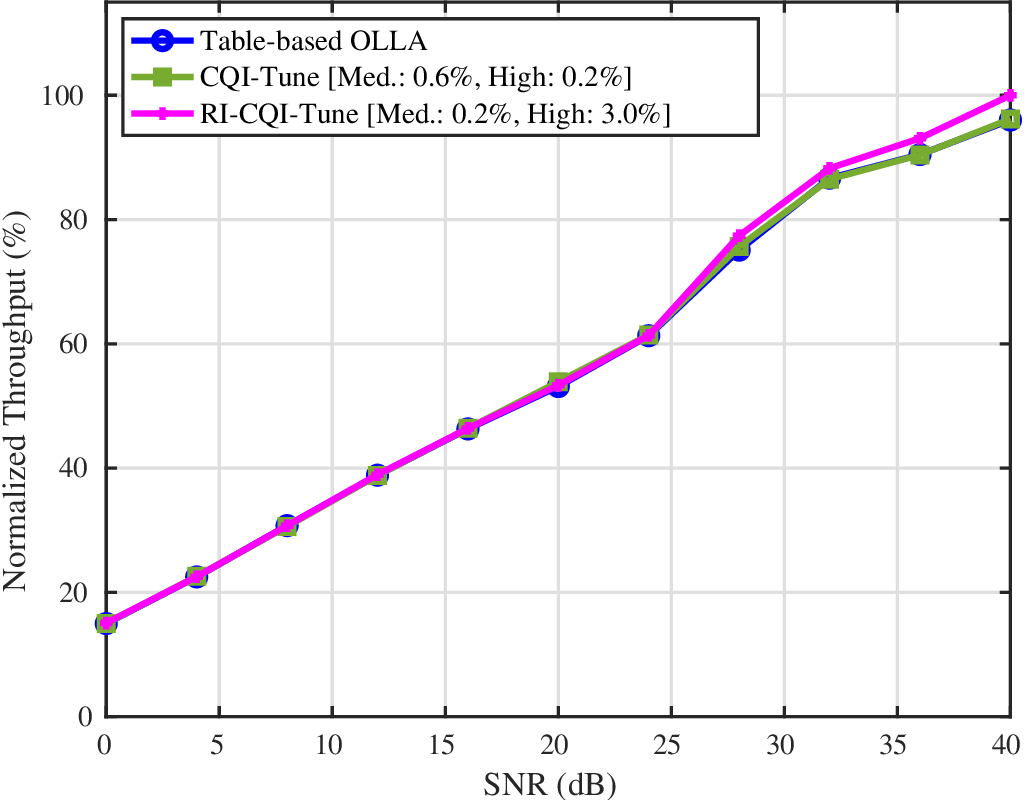}
			\caption{TDL-A10, Medium correlation \\ (Dop. freq. = 20 Hz).}
			\label{fig:tdla10_throughput_10ms_CSIRS_106RBs_med_corr}
		\end{subfigure}
		\hfill
		\begin{subfigure}{0.24\linewidth}
			\centering
			\includegraphics[width=\linewidth]{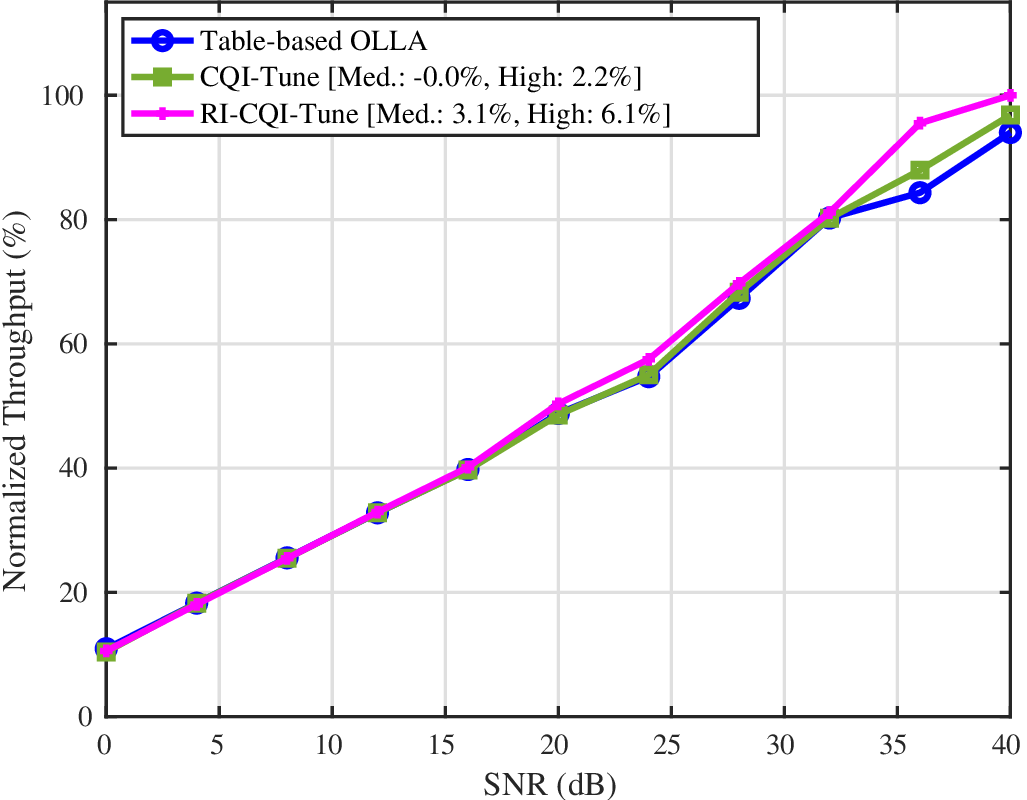}
			\caption{TDL-C200, Medium correlation \\ (Dop. freq. = 50 Hz).}
			\label{fig:tdlc200_throughput_10ms_CSIRS_106RBs_med_corr}
		\end{subfigure}
		\captionsetup{font=footnotesize}
		\caption{Throughput of \RICQI{} and \CQI{} with uniform sampling with CSI-RS period = 10 ms.}
		\label{fig:throughput_10ms_CSI_RS}
	\end{figure*}

	\begin{table*}
		\centering
		\captionsetup{font=footnotesize}
		\renewcommand{\arraystretch}{1.1}
		\begin{tabular}{@{}lllcccc@{}}
			\toprule
			\multirow{2}{*}{\textbf{CSI-RS Period}} & \multirow{2}{*}{\textbf{Channel Profile}} & \multirow{2}{*}{\textbf{Antenna Correlation}} & \multicolumn{2}{c}{\textbf{Medium SNR Gain}} & \multicolumn{2}{c}{\textbf{High SNR Gain}} \\
			\cmidrule(lr){4-5} \cmidrule(lr){6-7}
			& & & \textbf{CQI-Tune} & \textbf{RI-CQI-Tune} & \textbf{CQI-Tune} & \textbf{RI-CQI-Tune} \\
			\midrule
			\multirow{8}{*}{80 ms} 
			& TDL-A10, 20 Hz & Low & $5.3\%$ & $2.6\%$ & $1.3\%$ & $2.6\%$ \\
			& TDL-B50, 30 Hz & Low & $12.1\%$ & $8.1\%$ & $2\%$ & $1.1\%$ \\
			& TDL-B200, 50 Hz & Low & $7.0\%$ & $6.3\%$ & $0.7\%$ & $11.0\%$ \\
			& TDL-C200, 50 Hz & Low & $9.1\%$ & $8.5\%$ & $1.3\%$ & $10.9\%$ \\
			\cmidrule(lr){2-7}
			& TDL-B50, 30 Hz & Medium & $1.1\%$ & $1.7\%$ & $0.8\%$ & $0.9\%$ \\
			& TDL-B200, 50 Hz & Medium & $1.5\%$ & $6.2\%$ & $3.4\%$ & $15.5\%$ \\
			& TDL-B50, 30 Hz & High & $0.6\%$ & $1.4\%$ & $0.1\%$ & $2.1\%$ \\
			& TDL-B200, 50 Hz & High & $-0.6\%$ & $1.4\%$ & $0.7\%$ & $12.6\%$ \\
			\midrule
			\multirow{4}{*}{40 ms}
			& TDL-B50, 30 Hz & Medium & $0.3\%$ & $1.6\%$ & $0.4\%$ & $5.2\%$ \\
			& TDL-B50, 30 Hz & High & $0.5\%$ & $2.0\%$ & $2.4\%$ & $4.2\%$ \\
			& TDL-C200, 50 Hz & Low & $3.0\%$ & $6.7\%$ & $0.7\%$ & $9.1\%$ \\
			& TDL-C200, 50 Hz & High & $0.2\%$ & $5.2\%$ & $0.6\%$ & $11.5\%$ \\
			\midrule
			\multirow{4}{*}{10 ms}
			& TDL-B50, 30 Hz & Low & $2.0\%$ & $1.0\%$ & $0.2\%$ & $1.8\%$ \\
			& TDL-C200, 50 Hz & Low & $0.7\%$ & $0.1\%$ & $-0.2\%$ & $7.7\%$ \\
			& TDL-A10, 20 Hz & Medium & $0.6\%$ & $0.2\%$ & $0.2\%$ & $3.0\%$ \\
			& TDL-C200, 50 Hz & Medium & $0.0\%$ & $3.1\%$ & $2.2\%$ & $6.1\%$ \\
			\bottomrule
		\end{tabular}
		\caption{Throughput Gains of \CQI{} and \RICQI{} Across Various Scenarios}
		\label{table:consolidated_throughput_gains}
	\end{table*}
	
	\section{Conclusion}
	\label{sec:conclusion}
	
	This paper presented \name{}, a backpropagation-free, buffer-less online fine-tuning framework for on-device DT synchronization. \name{} opportunistically refines deployed ML models using delayed ground-truth feedback, achieving significant throughput gains over conventional OLLA baselines while remaining feasible on resource-constrained UE. These core principles extend naturally to backpropagation-based architectures. Extending \name{} to other 6G physical layer tasks, such as beam selection, is an area for future investigation.
	
	\bibliographystyle{IEEEtran}  
	\bibliography{references}  

@article{hinton2022forward,
	title={ {The Forward-Forward Algorithm: Some Preliminary Investigations}},
	author={Hinton, Geoffrey},
	journal={arXiv preprint arXiv:2212.13345},
	year={2022}
}

@article{lin1992self,
	title={ {Self-improving reactive agents based on reinforcement learning, planning and teaching}},
	author={Lin, Long-Ji},
	journal={Machine learning},
	volume={8},
	number={3},
	year={1992},
	publisher={Springer}
}

@misc{ali2026online,
	title={{Online Adaptation and ML-Non-ML Combining for Improved Wireless Link Adaptation}},
	author={Ali, Ramy E.  and Kwon, Hyukjoon},
	publisher={Google Patents},
	note={ {US} Patent, 2026}
}

@article{wiesmayr2025salad,
	title={SALAD: Self-adaptive link adaptation},
	author={Wiesmayr, Reinhard and Maggi, Lorenzo and Cammerer, Sebastian and Hoydis, Jakob and Aoudia, Fay{\c{c}}al A{\"\i}t and Keller, Alexander},
	journal={arXiv preprint arXiv:2510.05784},
	year={2025}
}

@article{schaul2015prioritized,
	title={ {Prioritized Experience Replay}},
	author={Schaul, Tom and others},
	journal={ ICLR},
	year={2016}
}

@article{rolnick2019experience,
	title={ {Experience Replay for Continual Learning}},
	author={Rolnick, David and others},
	journal={Advances in neural information processing systems},
	volume={32},
	year={2019}
}

@inproceedings{baknina2020adaptive,
	author       = {Abdulrahman Baknina and HyukJoon Kwon},
	title        = { {Adaptive {CQI} and {RI} Estimation for {5G NR}: A Shallow Reinforcement Learning Approach}},
	booktitle    = {IEEE Global Communications Conference (GLOBECOM)},
	year         = {2020},
}

@inproceedings{kingma2015adam,
	title={Adam: A Method for Stochastic Optimization},
	author={Kingma, Diederik P. and Ba, Jimmy},
	booktitle={International Conference on Learning Representations (ICLR)},
	year={2015}
}

@misc{forwardforward2023,
	author       = {Mohammad Pezeshki},
	title        = { {Implementation of Forward-Forward ({FF}) training algorithm}},
	year         = {2023},
	howpublished = {\url{https://github.com/mpezeshki/pytorch_forward_forward}},
}

@techreport{3gppTR38901,
	title        = { {Study on channel model for frequencies from 0.5 to 100 {GHz}}},
	number       = {TR 38.901 V14.0.0},
	year         = {2017},
	month        = {July},
	author = {3GPP}
}

@inproceedings{kaswan2024statistical,
	title={Statistical {AI/ML} model monitoring for {5G/6G}: Interference prediction case study},
	author={Kaswan, Priyanka and others},
	booktitle={IEEE International Conference on Communications Workshops (ICC Workshops)},
	year={2024},
}

@article{xu2024learning,
	title={ {Learning to estimate: A real-time online learning framework for MIMO-OFDM channel estimation}},
	author={Xu, Jiarui and others},
	journal={IEEE Trans. on Wireless Communications},
	year={2024},
	
}

@techreport{3gppTR38843,
	title        = { {Study on Artificial Intelligence (AI)/Machine Learning (ML) for NR Air Interface}},
	author       = {{3rd Generation Partnership Project (3GPP)}},
	institution  = {3GPP},
	type         = {Technical Report},
	number       = {38.843, Release 18},
	year         = {2023},
	
}

@article{xulearning,
	title={Learning at the speed of wireless: Online real-time learning for {AI-enabled} {MIMO} in {NextG}},
	author={Xu, Jiarui and others},
	journal={IEEE Comm. Magazine},
	year={2024},
}

@article{torres2025advancements,
	title={ {On Advancements of the Forward-Forward Algorithm}},
	author={Torres, Mauricio Ortiz and Lange, Markus and Raulf, Arne P},
	journal={arXiv preprint arXiv:2504.21662},
	year={2025}
}

@inproceedings{an2024dragon,
	title={ {DRAGON: A DRL-based MIMO Layer and MCS Adapter in Open RAN 5G Networks}},
	author={An, Qing and others},
	booktitle={Proceedings of the 30th Annual International Conference on Mobile Computing and Networking},
	year={2024}
}

@inproceedings{dong2018machine,
	title={Machine learning based link adaptation method for {MIMO} system},
	author={Dong, Zhijie and others},
	booktitle={IEEE 29th Annual International Symposium on Personal, Indoor and Mobile Radio Communications (PIMRC)},
	year={2018},
	
}

@article{saxena2021reinforcement,
	title={ {Reinforcement learning for efficient and tuning-free link adaptation}},
	author={Saxena, Vidit and Tullberg, Hugo and Jald{\'e}n, Joakim},
	journal={IEEE Transactions on Wireless Communications},
	volume={21},
	number={2},
	year={2021},
}

@article{huang2021deluxe,
	title={ {DELUXE: A DL-based link adaptation for URLLC/eMBB multiplexing in {5G} NR}},
	author={Huang, Yan and Hou, Y Thomas and Lou, Wenjing},
	journal={IEEE Journal on Selected Areas in Communications},
	volume={40},
	number={1},
	year={2021},
}

@article{van2021machine,
	title={ {Machine-learning-aided link-performance prediction for coded MIMO systems}},
	author={Van Le, Thuan and Lee, Kyungchun},
	journal={IEEE Transactions on Vehicular Technology},
	volume={71},
	number={3},
	year={2021},
}

@inproceedings{peralta2022outer,
	title={ {Outer loop link adaptation enhancements for ultra reliable low latency communications in 5G}},
	author={Peralta, Elena and others},
	booktitle={IEEE 95th Vehicular Technology Conference:(VTC-Spring)},
	year={2022},
}

@article{lecun1998gradient,
	title={Gradient-based learning applied to document recognition},
	author={LeCun, Yann and others},
	journal={Proc. IEEE},
	volume={86},
	number={11},
	year={1998}
}

@misc{Samsung2025AI6GR,
	author       = {Samsung},
	title        = { {AI/ML Use Cases and Framework for 6GR} },
	howpublished = {3GPP TSG RAN1 Meeting \#122, Bengaluru, India,  R1-2505588},
	month        = aug,
	year         = {2025},
}

@book{nesterov2013introductory,
	title={Introductory lectures on convex optimization: A basic course},
	author={Nesterov, Yurii},
	volume={87},
	year={2013},
	publisher={Springer Science \& Business Media}
}

@inproceedings{schroff2015facenet,
	title={Facenet: A unified embedding for face recognition and clustering},
	author={Schroff, Florian and Kalenichenko, Dmitry and Philbin, James},
	booktitle={Proceedings of the IEEE conference on computer vision and pattern recognition},
	year={2015}
}

@inproceedings{huang2025tinyfoa,
	title={ {TinyFoA}: Memory Efficient Forward-Only Algorithm for On-Device Learning},
	author={Huang, Baichuan and Aminifar, Amir},
	booktitle={Proceedings of the AAAI Conference on Artificial Intelligence},
	year={2025}
}

@inproceedings{de2023mu,
	title={{$\mu$-FF:} On-device forward-forward training algorithm for microcontrollers},
	author={De Vita, Fabrizio and others},
	booktitle={IEEE  Conference on Smart Computing},
	year={2023}
}

@book{horn2012matrix,
	title={Matrix analysis},
	author={Horn, Roger A and Johnson, Charles R},
	year={2012},
	publisher={Cambridge university press}
}

@inproceedings{Ali2026lighttune,
	title={{LightTune: Lightweight Online Fine-Tuning for 6G}},
	author={Ali, Ramy E. and Penna, Federico},
	booktitle={IEEE International Conference on Communications (ICC)},
	year={2026},
}

@inproceedings{karimi2016linear,
	title={Linear convergence of gradient and proximal-gradient methods under the polyak-{\l}ojasiewicz condition},
	author={Karimi, Hamed and Nutini, Julie and Schmidt, Mark},
	booktitle={Joint European conference on machine learning and knowledge discovery in databases},
	pages={795--811},
	year={2016},
	organization={Springer}
}

@article{pinsker1964information,
	title={Information and information stability of random variables and processes},
	author={Pinsker, Mark S},
	journal={Holden-Day},
	year={1964}
}

@article{mazumdar2026enhancing,
	title={Enhancing OLLA via Exponential Decay for Efficient Link Adaptation in Emerging 6G Traffic},
	author={Mazumdar, Aritra and Paris, Stefano and Amiri, Abolfazl and Pedersen, Klaus I and Adeogun, Ramoni},
	journal={IEEE Access},
	volume={14},
	pages={5764--5776},
	year={2026},
	publisher={IEEE}
}

@article{maggi2026sinr,
	title={SINR Estimation under Limited Feedback via Online Convex Optimization},
	author={Maggi, Lorenzo and Bonev, Boris and Wiesmayr, Reinhard and Cammerer, Sebastian and Keller, Alexander},
	journal={arXiv preprint arXiv:2603.02061},
	year={2026}
}

@techreport{3gpp_ts38214,
	author = {3GPP},
	title  = {Physical layer procedures for data (Release 16)},
	type   = {Technical Specification (TS)},
	number = {38.214},
	year   = {2021}
}

@article{zhang2025digital,
	title={Digital network twins for next-generation wireless: Creation, optimization, and challenges},
	author={Zhang, Zifan and Peng, Zhiyuan and Yu, Hanzhi and Chen, Mingzhe and Liu, Yuchen},
	journal={IEEE network},
	year={2025},
	publisher={IEEE}
}

@article{tong2025continual,
	title={Continual reinforcement learning for digital twin synchronization optimization},
	author={Tong, Haonan and Chen, Mingzhe and Zhao, Jun and Hu, Ye and Yang, Zhaohui and Liu, Yuchen and Yin, Changchuan},
	journal={IEEE transactions on mobile computing},
	year={2025},
	publisher={IEEE}
}

@article{lin20236g,
	title={6G digital twin networks: From theory to practice},
	author={Lin, Xingqin and Kundu, Lopamudra and Dick, Chris and Obiodu, Emeka and Mostak, Todd and Flaxman, Mike},
	journal={IEEE Comm. Magazine},
	volume={61},
	number={11},
	pages={72--78},
	year={2023},
	publisher={IEEE}
}

@article{luo2025digital,
	title={Digital Twin Aided Massive MIMO CSI Feedback: Exploring the Impact of Twinning Fidelity},
	author={Luo, Hao and Jiang, Shuaifeng and Khosravirad, Saeed R and Alkhateeb, Ahmed},
	journal={IEEE Transactions on Communications},
	year={2025},
	publisher={IEEE}
}
	
	\appendices
	\vspace{-10 pt}
	\section{Validation of the Proposed Loss Function}
	\label{app:loss_experiment}
	We derive our alternative quadratic loss from the second-order Taylor expansion of the function $f(x) = \ln(1 + e^x)$, centered at $x = 0$ that is given as:
	$
	f(x) = \ln 2 + \frac{1}{2}x + \frac{1}{8}x^2 + R_2(x),
	$
	where $R_2(x)$ denotes the Lagrange remainder term, defined as $R_2(x) = \frac{f^{(3)}(\xi)}{3!}x^3 = \frac{1}{6} \frac{e^\xi (1 - e^\xi)}{(1 + e^\xi)^3} x^3$ for some $\xi \in [0, x]$. We apply this expansion to the Softplus loss while discarding the higher-order remainder $R_2(x)$. Because the optimization landscape is invariant to constant shifts, we discard the constant $\ln 2$ term. Finally, to eliminate the fractional coefficients and simplify the hardware arithmetic, we apply a uniform scaling factor of $8$ to get our proposed loss function.

To validate the proposed loss, we compare the test accuracy with that of the Softplus loss  on the MNIST dataset \cite{lecun1998gradient}. The code is available  in~\cite{forwardforward2023}. Our results in Table~\ref{table:mnist_accuracy_loss} demonstrate that the proposed loss function achieves comparable accuracy while offering enhanced computational efficiency.

	\begin{table}[htb!]
		\centering
		\begin{tabular}{lcc}
			\toprule
			\textbf{Loss Function} & \textbf{Softplus} & \textbf{Proposed} \\
			\midrule
			Test Accuracy & 93.15\% & 93.74\% \\
			\bottomrule
		\end{tabular}
		\captionsetup{font=footnotesize}
		\caption{Test accuracy on MNIST using Softplus loss vs. proposed loss }
		\label{table:mnist_accuracy_loss}
	\end{table}

	\section{Convergence Analysis of LightTune}
	\label{app:convergence}
	
	This appendix provides the convergence proof. All random variables are defined on a common probability space $(\Omega,\mathcal{F},\mathbb{P})$. The initial parameters $\bm{\theta}^{(1)}$ are from offline training on $\mathcal{D}_1$. Online data $\{\bm{z}^{(t)}_{\text +}\}_{t=1}^{\infty}$ are i.i.d. from $\mathcal{D}_2$. At each $t$, a negative label $y^{(t)}_{\text{-}}$ is drawn from a distribution $R(\cdot\mid y^{(t)}_{\text{+}})$ over $\mathcal{Y}\setminus\{y^{(t)}_{\text{+}}\}$. We write $\mathbb{E}_{y^{(t)}_{\text{-}}}[\cdot \mid \mathcal{F}_t]$ for expectation over $y^{(t)}_{\text{-}}$ given the past.

	\subsection{Justification of Assumptions}
	Assumption \ref{assumption:gradient_lower_bound} is motivated by the observation that a vanishing expected squared norm would necessitate a perfect cancellation between the positive and negative gradients in expectation. Such systematic cancellation is highly improbable as long as the model is driven by non-negligible prediction errors. In particular, if the negative label $y^{(t)}_{\text{-}}$ is drawn from a distribution that assigns positive probability to every incorrect label (e.g.,  uniform), then $\mathbb{E}_{y^{(t)}_{\text{-}}}\bigl[\|\nabla\mathcal{L}_L^{(t)}\|_2^2\bigr] = 0 \implies \|\bm{p} + \bm{n}_c\|_2 = 0$ for every $c \neq y^{(t)}_{\text{+}}$, where $\bm{p} = \nabla\mathcal{L}_{\text{+},l}^{(t)}$ and $\bm{n}_c = \nabla\mathcal{L}_{\text{-},l}^{(t)}$ for label $c$. Hence $\bm{n}_c = -\bm{p}$ for all $c$. This would mean that the gradient does not depend on the label $c$, implying that the network’s response to different incorrect labels is identical in the sense of its gradient. When the prediction error is large, the true label and the predicted label differ, consequently, the gradients for those labels cannot all be the same. Therefore, for any reasonable sampling scheme that covers all incorrect labels (e.g., uniform), the expected squared gradient norm must be bounded below by a positive constant $\gamma_2(\delta)>0$. 
	
	\subsection{Preliminary Lemmas}
	
	We recall and introduce some notation.  The parameters of neuron $j$ in layer $l$ are $\bm{\theta}_{l,j} = [\bm{w}_{l,j}, b_{l,j}]^\top$. Hence, the positive pre‑activation is given by $p_{\text{+},l}^{(t)}[j] = \bm{\theta}_{l,j}^\top \tilde{\bm{h}}_{\text{+},l-1}^{(t)}$, and the negative pre‑activation is given by $p_{\text{-},l}^{(t)}[j] = \bm{\theta}_{l,j}^\top \tilde{\bm{h}}_{\text{-},l-1}^{(t)}$. From \eqref{eqn:proposed_loss}, the per-neuron loss combines positive and negative contributions as $\mathcal{L}_{l}^{(t)}[j](\bm{\theta}_{l,j}) = ((h_{\text{+},l}^{(t)}[j])^2 - T)^2 - 4((h_{\text{+},l}^{(t)}[j])^2 - T) + ((h_{\text{-},l}^{(t)}[j])^2 - T)^2 + 4((h_{\text{-},l}^{(t)}[j])^2 - T)$, with the total layer loss $\mathcal{L}_l^{(t)}(\bm{\theta}) = \frac{1}{M_l}\sum_{j=1}^{M_l} \mathcal{L}_{l}^{(t)}[j](\bm{\theta}_{l,j})$.\\
	We begin with the proof of Lemma \ref{lemma:bounded_activations}.
	
	\begin{proof}
		We prove by induction on the layer index $l$. \textbf{Base case $l = 0$}: By definition, $\bm{h}_0^{(t)} = \bm{z}^{(t)}$, and Assumption~\ref{assumption:boundedness}(i) gives $\|\bm{z}^{(t)}\|_2 \le B_z$. Hence $\|\bm{h}_0^{(t)}\|_2 \le B_z =: B_0$. \textbf{Inductive step}: Assume that for some $l \geq 1$, we have $\|\bm{h}_{l-1}^{(t)}\|_2 \le B_{l-1}$. Consider layer $l$. For any neuron $j$ in this layer, the pre‑activation is $p_{l}^{(t)}[j] = \bm{\theta}_{l,j}^{\top} \tilde{\bm{h}}_{l-1}^{(t)}$, where $\tilde{\bm{h}}_{l-1}^{(t)} = [\bm{h}_{l-1}^{(t)}, 1]^\top$ is the augmented input. By the Cauchy–Schwarz inequality and Assumption~\ref{assumption:boundedness}(ii), $|p_{l}^{(t)}[j]| \le \|\bm{\theta}_{l,j}\|_2 \, \|\tilde{\bm{h}}_{l-1}^{(t)}\|_2 \le B_\theta (\|\bm{h}_{l-1}^{(t)}\|_2 + 1) \le B_\theta (B_{l-1}+1)$. The activation is $h_{l}^{(t)}[j] = \max(0, p_{l}^{(t)}[j])$, so $|h_{l}^{(t)}[j]| \le |p_{l}^{(t)}[j]|$. Therefore, 
		$
		\|\bm{h}_{l}^{(t)}\|_2 = \sqrt{ \sum_{j=1}^{M_l} (h_{l}^{(t)}[j])^2 } \le \sqrt{M_l}\, B_\theta (B_{l-1}+1) =: B_l. 
		$
		This completes the induction. Finally, since the network has a finite number of layers $L$, the set $\{B_0, B_1, \dots, B_L\}$ is finite. Taking $B_h = \max_{0\le l\le L} B_l$ gives a uniform bound valid for all layers and all time steps.
	\end{proof}
	
	Next, we provide the proof of Lemma \ref{lemma:bounded_gradient} which shows that the gradients are also bounded. 		
	\begin{proof}
		The gradient for neuron $j$ is given as $\nabla_{\bm{\theta}_{l,j}} \mathcal{L}_l^{(t)}(\bm{\theta}_l) = \frac{1}{M_l}\Bigl( \nabla_{\bm{\theta}_{l,j}} \mathcal{L}_{\text{+},l}^{(t)}[j] + \nabla_{\bm{\theta}_{l,j}} \mathcal{L}_{\text{-},l}^{(t)}[j] \Bigr)$, where all quantities on the right-hand side are evaluated at the current parameters $\bm{\theta}_{l,j}^{(t)}$ and the fixed inputs from the previous layer. We first bound each branch separately using the bounds from Assumption \ref{assumption:boundedness} and Lemma \ref{lemma:bounded_activations}. For the positive branch, the derivative of the scalar loss component is $2(g_{\text{+},l}^{(t)}[j]-T) - 4 = 2(g_{\text{+},l}^{(t)}[j]-T-2)$. Applying the chain rule with $g_{\text{+},l}^{(t)}[j] = (h_{\text{+},l}^{(t)}[j])^2$ yields an additional factor of $2h_{\text{+},l}^{(t)}[j]$. Thus, we can bound each term in the gradient: $|4((h_{\text{+},l}^{(t)}[j])^2 - T - 2)| \le 4(B_h^2 + T + 2)$, $|h_{\text{+},l}^{(t)}[j]| \le B_h$, $\|\tilde{\bm{h}}_{\text{+},l-1}^{(t)}\|_2 \le B_h + 1$, and $|\mathds{1}\{p_{\text{+},l}^{(t)}[j] > 0\}| \le 1$. Multiplying these bounds yields
		\begin{equation}
			\|\nabla_{\bm{\theta}_{l,j}} \mathcal{L}_{\text{+},l}^{(t)}[j]\|_2 \le 4(B_h^2 + T + 2) B_h (B_h + 1).
		\end{equation} 
		Applying the same reasoning to the negative branch gives the identical bound: 
		$
		\|\nabla_{\bm{\theta}_{l,j}} \mathcal{L}_{\text{-},l}^{(t)}[j]\|_2 \le 4(B_h^2 + T + 2) B_h (B_h + 1).
		$
		Finally, using the triangle inequality and the expression for the total gradient evaluated at $\bm{\theta}_l^{(t)}$, we obtain:
		\begin{align}
			\|\nabla_{\bm{\theta}_{l,j}} \mathcal{L}_l^{(t)}(\bm{\theta}_l^{(t)})\|_2 &\le \frac{1}{M_l}\Bigl( \|\nabla_{\bm{\theta}_{l,j}} \mathcal{L}_{\text{+},l}^{(t)}[j]\|_2 + \|\nabla_{\bm{\theta}_{l,j}} \mathcal{L}_{\text{-},l}^{(t)}[j]\|_2 \Bigr) \notag \\
			&\leq \frac{8(B_h^2 + T + 2) B_h (B_h + 1)}{M_l}.
		\end{align}
	\end{proof}

	Next, we prove Lemma \ref{lemma:bounded_loss} showing that our  loss  is bounded. 

	\begin{proof}
		Recall that for any neuron, $g = h^2$. From Lemma \ref{lemma:bounded_activations}, we have $|h| \le B_h$, so $0 \le g \le B_h^2$. Since $g$ is bounded by $B_h^2$, we have $|(g - T)^2 - 4(g - T)| \le (B_h^2 + T)^2 + 4(B_h^2 + T) \le (B_h^2 + T + 2)^2 + 4(B_h^2 + T + 2)$, and similarly for the negative branch term. Adding the two bounds and noting that the loss for the layer is the average over neurons, we obtain:
		\begin{equation*}
			|\mathcal{L}_L^{(t)}(\bm{\theta}_L)| \le (B_h^2 + T + 2)^2 + 4(B_h^2 + T + 2) =: M.
		\end{equation*}
	\end{proof}
	Finally, we prove Lemma \ref{lemma:smoothness}, which demonstrates that the proposed loss function is smooth. 
	\begin{proof}
		Let $d_l$ be the total number of parameters in layer $l$. The parameter vector $\bm{\theta}_l \in \mathbb{R}^{d_l}$ can be written as the concatenation of the parameter vectors for each neuron as $\bm{\theta}_l = [\bm{\theta}_{l,1}^\top, \bm{\theta}_{l,2}^\top, \dots, \bm{\theta}_{l,M_l}^\top]^\top \in \mathbb{R}^{d_l}$, where $\bm{\theta}_{l,j} \in \mathbb{R}^{d_n}$ are the parameters of neuron $j$. The loss function for layer $l$ is the average over neurons: 
		$
		\mathcal{L}_l^{(t)}(\bm{\theta}_l) = \frac{1}{M_l}\sum_{j=1}^{M_l} \mathcal{L}_{l}^{(t)}[j](\bm{\theta}_{l,j}),
		$
		where each $\mathcal{L}_{l}^{(t)}[j](\bm{\theta}_{l,j})$ depends only on $\bm{\theta}_{l,j}$ and not on the parameters of other neurons.
		
		\paragraph{Gradient structure.} Because the loss separates over neurons, the gradient with respect to $\bm{\theta}_l$ is the concatenation of the per‑neuron gradients: 
		$
		\nabla \mathcal{L}_l^{(t)}(\bm{\theta}_l) = \frac{1}{M_l}[\nabla_{\bm{\theta}_{l,1}} \mathcal{L}_{l}^{(t)}[1]^\top, \dots, \nabla_{\bm{\theta}_{l,M_l}} \mathcal{L}_{l}^{(t)}[M_l]^\top]^\top.
		$
		
		\paragraph{Hessian structure.} Differentiating again, the Hessian matrix $\nabla^2 \mathcal{L}_l^{(t)}(\bm{\theta}_l) \in \mathbb{R}^{d_l \times d_l}$ is block‑diagonal: 
		\begin{equation*}
		\nabla^2 \mathcal{L}_l^{(t)}(\bm{\theta}_l) = \frac{1}{M_l} \text{diag} \left( \nabla^2_{\bm{\theta}_{l,1}} \mathcal{L}_{l}^{(t)}[1], \dots, \nabla^2_{\bm{\theta}_{l,M_l}} \mathcal{L}_{l}^{(t)}[M_l] \right). 
	\end{equation*}
		This is because cross‑partial derivatives $\frac{\partial^2}{\partial \bm{\theta}_{l,i} \partial \bm{\theta}_{l,j}}$ for $i \neq j$ are zero, as each term in the sum depends only on its own neuron's parameters.
		
		\paragraph{Spectral norm of a block‑diagonal matrix.} For a block‑diagonal matrix, the spectral norm equals the maximum of the spectral norms of the individual blocks \cite{horn2012matrix}. Therefore, 
		\begin{equation}
			\|\nabla^2 \mathcal{L}_l^{(t)}(\bm{\theta}_l)\|_2 = \frac{1}{M_l} \max_{1 \le j \le M_l} \bigl\|\nabla^2_{\bm{\theta}_{l,j}} \mathcal{L}_{l}^{(t)}[j](\bm{\theta}_{l,j})\bigr\|_2.
		\end{equation} 
		Thus it suffices to bound the Hessian of a single neuron; the full Hessian norm will be at most that bound divided by $M_l$.
		
		\paragraph{Bounding the per‑neuron Hessian.} Fix a neuron $j$ and time $t$, and drop the indices $l,j,t$ for brevity. When $p_{\text +} > 0$ (neuron active), we have 
		$
		\nabla \mathcal{L}_{\text +} = \bigl[4p_{\text +}^3 - 4(T+2)p_{\text +}\bigr] \tilde{\bm{h}}_{\text +}.
		$
		Differentiating again with respect to $\bm{\theta}$ yields the Hessian
		\begin{align}
			\nabla^2 \mathcal{L}_{\text +} = \bigl[12p_{\text +}^2 - 4(T+2)\bigr] \tilde{\bm{h}}_{\text +} \tilde{\bm{h}}_{\text +}^{\top}.
		\end{align}
		When $p_{\text +} \le 0$, the neuron is inactive and the Hessian is zero. Similarly, when $p_{\text -} > 0$, $\nabla^2 \mathcal{L}_{\text -} = \bigl[12p_{\text -}^2 - 4(T+2)\bigr] \tilde{\bm{h}}_{\text -} \tilde{\bm{h}}_{\text -}^{\top}$, and zero otherwise. By Lemma \ref{lemma:bounded_activations}, $\|\tilde{\bm{h}}_{\text +}\|_2 \le B_h+1$ and $\|\tilde{\bm{h}}_{\text -}\|_2 \le B_h+1$. Hence, we have 
		\begin{equation}
			\|\nabla^2 \mathcal{L}_{\text +}\|_2 \le (12B_h^2 + 4T + 8)(B_h+1)^2.
		\end{equation} 
		The same bound holds for $\|\nabla^2 \mathcal{L}_{\text -}\|_2$. By the triangle inequality, we have
		\begin{align}
			&\|\nabla^2 \mathcal{L}_{\text +} + \nabla^2 \mathcal{L}_{\text -}\|_2 \le 8(3B_h^2+T+2)(B_h+1)^2 := \bar{\rho}_l. 
		\end{align}
		Thus we have shown that for any neuron $j$,
		$
			\|\nabla^2_{\bm{\theta}_{l,j}} \mathcal{L}_{l}^{(t)}[j](\bm{\theta}_{l,j})\|_2 \le \bar{\rho}_l .
		$
		Using the block‑diagonal structure and the fact that the spectral norm of a block‑diagonal matrix is the maximum of the block norms scaled by the overall factor $1/M_l$, we obtain:
		\begin{equation}
			\|\nabla^2 \mathcal{L}_l^{(t)}(\bm{\theta}_l)\|_2 \le \frac{1}{M_l} \max_{j} \|\nabla^2_{\bm{\theta}_{l,j}} \mathcal{L}_{l}^{(t)}[j]\|_2 \le \frac{\bar{\rho}_l}{M_l} =: \rho_l.
		\end{equation}

		\paragraph{From Hessian bound to Lipschitz gradient.} Now we prove that this bound on the Hessian implies the gradient is $\rho_l$-Lipschitz everywhere. Take any $\bm{\theta}_l, \bm{\theta}'_l \in \mathbb{R}^{d_l}$ and consider the line segment $\bm{\theta}_l(s) = \bm{\theta}_l + s(\bm{\theta}'_l - \bm{\theta}_l),$ where $s \in [0,1]$.\\
		1) \textbf{Points where the gradient may not be differentiable.} \\
		For each neuron $k$ in layer $l$, its pre‑activation along the segment is $p_k(s) = \bm{\theta}_{l,k}(s)^\top \tilde{\bm{h}}_{l-1}$, where $\bm{\theta}_{l,k}(s)$ is the part of $\bm{\theta}(s)$ corresponding to neuron $k$, and $\tilde{\bm{h}}_{l-1}$ is fixed. This is an affine function of $s$, i.e., $p_k(s) = a_k s + b_k$ for some constants $a_k,b_k$.\\			
		2) \textbf{Zeros of affine functions are isolated.}
		For a fixed $k$, the equation $p_k(s)=0$ is linear in $s$. Hence it has either: 1) no solution (if $a_k = 0$ and $b_k \neq 0$), 2) exactly one solution $s_k^*$ (if $a_k \neq 0$) or 3) the whole interval (if $a_k = 0$ and $b_k = 0$, which  means the pre‑activation is identically 0; this degenerate case occurs on a set of measure zero and can be ignored). Thus each neuron contributes at most one point where $p_k(s)=0$.\\
		3) \textbf{The exceptional set is finite.}
		Since there are finitely many neurons, the set $S_0 = \{ s \in [0,1] : \exists k \text{ such that } p_k(s) = 0 \}$ is finite. Order its elements as $0 \le s_1  < \dots < s_m \le 1$. Remove these points to obtain a partition of $[0,1]$ into subintervals $[0, s_1],\; [s_1, s_2],\; \dots,\; [s_m, 1]$. On each such subinterval, no pre‑activation changes sign, so the activation pattern (which neurons are active) remains fixed. Consequently, on each subinterval, the gradient $\nabla\mathcal{L}_l^{(t)}(\bm{\theta}(s))$ is a polynomial in $s$ (because the per‑neuron contributions are polynomials in the parameters, and the parameters are affine in $s$). Hence it is infinitely differentiable on the open interval and extends continuously to the endpoints. \\		
		4) \textbf{Derivative on a smooth subinterval.}
		On any subinterval where $\nabla\mathcal{L}_l^{(t)}(\bm{\theta}(s))$ is $C^1$, we can differentiate: $\frac{d}{ds} \nabla\mathcal{L}_l^{(t)}(\bm{\theta}_l(s)) = \nabla^2\mathcal{L}_l^{(t)}(\bm{\theta}_l(s)) \, (\bm{\theta}_l' - \bm{\theta}_l)$, where the Hessian exists everywhere on the interval because the activation pattern is constant. From the bound on the Hessian, we have $\|\nabla^2\mathcal{L}_l^{(t)}(\bm{\theta}_l(s))\|_2 \le \rho_l$, so 
		$
		\left\|\frac{d}{ds} \nabla\mathcal{L}_l^{(t)}(\bm{\theta}_l(s))\right\|_2 \le \rho_l \|\bm{\theta}'_l - \bm{\theta}_l\|_2.
		$
		\noindent 		
		5) \textbf{Integration over each subinterval.}
		Apply the fundamental theorem of calculus on each subinterval. Because $\nabla\mathcal{L}_l^{(t)}(\bm{\theta}(s))$ is continuously differentiable on the open interval and continuous up to the endpoints, we have $\nabla\mathcal{L}_l^{(t)}(\bm{\theta}_l(s_{i+1})) - \nabla\mathcal{L}_l^{(t)}(\bm{\theta}_l(s_i)) = \int_{s_i}^{s_{i+1}} \nabla^2\mathcal{L}_l^{(t)}(\bm{\theta}_l(s)) \, (\bm{\theta}_l' - \bm{\theta}_l) \, ds$. Summing these equalities from $i=0$ to $m$ (with $s_0 = 0$, $s_{m+1}=1$) telescopes the left‑hand side, giving:
		\begin{equation*}
			\nabla\mathcal{L}_l^{(t)}(\bm{\theta}_l') - \nabla\mathcal{L}_l^{(t)}(\bm{\theta}_l) = \sum_{i=0}^{m} \int_{s_i}^{s_{i+1}} \nabla^2\mathcal{L}_l^{(t)}(\bm{\theta}_l(s)) \, (\bm{\theta}'_l - \bm{\theta}_l) \, ds.
		\end{equation*}
		
		Taking norms and using the triangle inequality yields 
		\begin{align}
			&\|\nabla\mathcal{L}_l^{(t)}(\bm{\theta}_l') - \nabla\mathcal{L}_l^{(t)}(\bm{\theta}_l)\|_2 \notag \\ &\le \sum_{i=0}^{m} \int_{s_i}^{s_{i+1}} \|\nabla^2\mathcal{L}_l^{(t)}(\bm{\theta}_l(s))\|_2 \, \|\bm{\theta}_l' - \bm{\theta}_l\|_2 \, ds \notag \\
			&\le \rho_l \|\bm{\theta}_l' - \bm{\theta}_l\|_2 \sum_{i=0}^{m} (s_{i+1} - s_i) 
			= \rho_l \|\bm{\theta}_l' - \bm{\theta}_l\|_2.
		\end{align}
		Thus, $\mathcal{L}_l^{(t)}$ is $\rho_l$-smooth.
	\end{proof}
	
	\subsection{Convergence Theorem}
	We now provide the proof of Theorem \ref{thm:convergence}.
	
	\begin{proof}
		We proceed in steps as follows. 
		
		\begin{enumerate}[leftmargin=*, itemsep=10pt, parsep=0pt]
			\item \textbf{Local decrease.} For any $t$, if $I^{(t)}_\delta = 1$, the algorithm performs a gradient update: $\bm{\theta}_L^{(t+1)} = \bm{\theta}_L^{(t)} - \alpha_f \nabla \mathcal{L}_L^{(t)}(\bm{\theta}_L^{(t)})$. Because $\mathcal{L}_L^{(t)}$ is $\rho_L$-smooth (Lemma \ref{lemma:smoothness}), we  apply the descent lemma (Lemma \ref{lemma:descent}) with $\bm{\theta} = \bm{\theta}_L^{(t)}$ and $\bm{\theta}' = \bm{\theta}_L^{(t+1)}$ as follows 
			\begin{align}
				\mathcal{L}_L^{(t)}(\bm{\theta}_L^{(t+1)}) &\le \mathcal{L}_L^{(t)}(\bm{\theta}_L^{(t)}) + \nabla \mathcal{L}_L^{(t)}(\bm{\theta}_L^{(t)})^\top (\bm{\theta}_L^{(t+1)} - \bm{\theta}_L^{(t)}) \notag \\ &+ \frac{\rho_L}{2} \|\bm{\theta}_L^{(t+1)} - \bm{\theta}_L^{(t)}\|_2^2.
			\end{align} 
			Substituting the update $\bm{\theta}_L^{(t+1)} - \bm{\theta}_L^{(t)} = -\alpha_f \nabla \mathcal{L}_L^{(t)}(\bm{\theta}_L^{(t)})$ yields
			\begin{equation*}
				\mathcal{L}_L^{(t)}(\bm{\theta}_L^{(t+1)}) \le \mathcal{L}_L^{(t)}(\bm{\theta}_L^{(t)}) - \alpha_f\left(1 - \frac{\rho_L \alpha_f}{2}\right) \|\nabla \mathcal{L}_L^{(t)}(\bm{\theta}_L^{(t)})\|_2^2.
			\end{equation*} 
			Since $\alpha_f < 1/\rho_L$, we have $1 - \frac{\rho_L \alpha_f}{2} > \frac12$. Therefore, $\mathcal{L}_L^{(t)}(\bm{\theta}_L^{(t+1)}) \le \mathcal{L}_L^{(t)}(\bm{\theta}_L^{(t)}) - \frac{\alpha_f}{2} \|\nabla \mathcal{L}_L^{(t)}(\bm{\theta}_L^{(t)})\|_2^2$. If $I^{(t)}_\delta = 0$, no update occurs, so $\mathcal{L}_L^{(t)}(\bm{\theta}_L^{(t+1)}) = \mathcal{L}_L^{(t)}(\bm{\theta}_L^{(t)})$. Combining both cases yields:
			\begin{equation}
				\mathcal{L}_L^{(t)}(\bm{\theta}_L^{(t+1)}) \le \mathcal{L}_L^{(t)}(\bm{\theta}_L^{(t)}) - \frac{\alpha_f}{2}\|\nabla\mathcal{L}_L^{(t)}(\bm{\theta}_L^{(t)})\|_2^2 I^{(t)}_\delta.
			\end{equation}
			
			\item \textbf{Conditional expectation under $\mathcal{D}_2$.} Conditioning on $\mathcal{F}_t$ (which fixes $\bm{\theta}_L^{(t)}$, $\bm{x}^{(t)}$, $y^{(t)}_{\text +}$) and using the gradient lower bound (Assumption \ref{assumption:gradient_lower_bound}) yields
			\begin{align}
				&\mathbb{E}_{y^{(t)}_{\text{-}}}\bigl[\mathcal{L}_L^{(t)}(\bm{\theta}_L^{(t+1)}) \mid \mathcal{F}_t\bigr] \le \mathcal{L}_L^{(t)}(\bm{\theta}_L^{(t)}) \notag \\
				&- \frac{\alpha_f}{2} I^{(t)}_\delta \mathbb{E}_{y^{(t)}_{\text{-}}}\bigl[\|\nabla\mathcal{L}_L^{(t)}(\bm{\theta}_L^{(t)})\|_2^2 \mid \mathcal{F}_t, I^{(t)}_\delta=1\bigr] \notag \\
				&\le \mathcal{L}_L^{(t)}(\bm{\theta}_L^{(t)}) - \frac{\alpha_f \gamma_2(\delta)}{2} I^{(t)}_\delta.
			\end{align}
			
			\item \textbf{Total expectation under $\mathcal{D}_2$.} Taking expectation under $\mathcal{D}_2$:
			\begin{equation*}
				\mathbb{E}_{\mathcal{D}_2}[\mathcal{L}_L^{(t)}(\bm{\theta}_L^{(t+1)})] \le \mathbb{E}_{\mathcal{D}_2}[\mathcal{L}_L^{(t)}(\bm{\theta}_L^{(t)})] - \frac{\alpha_f \gamma_2(\delta)}{2} \mathbb{E}_{\mathcal{D}_2}[I^{(t)}_\delta].
			\end{equation*}
			
			\item \textbf{Relating to $\mathcal{D}_1$ via Pinsker.} For any fixed $\bm{\theta}$, by Lemma \ref{lemma:pinsker} applied with $P=\mathcal{D}_2$, $Q=\mathcal{D}_1$, and $f=\mathcal{L}_L^{(t)}(\bm{\theta})$, $|\mathbb{E}_{\mathcal{D}_2}[\mathcal{L}_L^{(t)}(\bm{\theta})] - \mathbb{E}_{\mathcal{D}_1}[\mathcal{L}_L^{(t)}(\bm{\theta})]| \le M \sqrt{\frac{1}{2} D_{\mathrm{KL}}(\mathcal{D}_2\|\mathcal{D}_1)}$. Since $\bm{\theta}_L^{(t)}$ is independent of the sample at time $t$, we can condition on $\bm{\theta}_L^{(t)}$ and integrate: 
			\begin{align}
				&\mathbb{E}_{\mathcal{D}_2}[\mathcal{L}_L^{(t)}(\bm{\theta}_L^{(t)})] = \mathbb{E}\bigl[ \mathbb{E}_{\mathcal{D}_2}[\mathcal{L}_L^{(t)}(\bm{\theta}) \mid \bm{\theta}=\bm{\theta}_L^{(t)}] \bigr] \notag \\
				&\le \mathbb{E}\bigl[ \mathbb{E}_{\mathcal{D}_1}[\mathcal{L}_L^{(t)}(\bm{\theta}) \mid \bm{\theta}=\bm{\theta}_L^{(t)}] + M\sqrt{\tfrac{1}{2}D_{\mathrm{KL}}} \bigr] \notag \\
				&= \mathbb{E}_{\mathcal{D}_1}[\mathcal{L}_L^{(t)}(\bm{\theta}_L^{(t)})] + M\sqrt{\tfrac{1}{2}D_{\mathrm{KL}}}.
			\end{align} 
			By Assumption \ref{assumption:distributions}, the offline training data are i.i.d. from $\mathcal{D}_1$, so for any fixed $\bm{\theta}$, $\mathbb{E}_{\mathcal{D}_1}[\mathcal{L}_L^{(t)}(\bm{\theta})] = \mathbb{E}_{\mathcal{D}_1}[\mathcal{L}_L^{(1)}(\bm{\theta})]$. Moreover, because $\bm{\theta}_L^{(t)}$ is independent of the sample at time $t$ under $\mathcal{D}_1$ as well: 
			\begin{align}
				\mathbb{E}_{\mathcal{D}_1}[\mathcal{L}_L^{(t)}(\bm{\theta}_L^{(t)})] &= \mathbb{E}\bigl[ \mathbb{E}_{\mathcal{D}_1}[\mathcal{L}_L^{(1)}(\bm{\theta}) \mid \bm{\theta}=\bm{\theta}_L^{(t)}] \bigr] \notag \\ &= \mathbb{E}_{\mathcal{D}_1}[\mathcal{L}_L^{(1)}(\bm{\theta}_L^{(t)})].
			\end{align} 
			Thus, $\mathbb{E}_{\mathcal{D}_2}[\mathcal{L}_L^{(t)}(\bm{\theta}_L^{(t)})] \le \mathbb{E}_{\mathcal{D}_1}[\mathcal{L}_L^{(1)}(\bm{\theta}_L^{(t)})] + M\sqrt{\tfrac{1}{2}D_{\mathrm{KL}}}$. The same inequality holds for $\bm{\theta}_L^{(t+1)}$. Substituting these upper bounds yields:
			\begin{equation}
				\mathbb{E}_{\mathcal{D}_1}[\mathcal{L}_L^{(1)}(\bm{\theta}_L^{(t+1)})] \le \mathbb{E}_{\mathcal{D}_1}[\mathcal{L}_L^{(1)}(\bm{\theta}_L^{(t)})] - \frac{\alpha_f \gamma_2(\delta)}{2} \mathbb{E}_{\mathcal{D}_2}[I^{(t)}_\delta].
			\end{equation}
			
			\item \textbf{Telescoping.} Summing from $t=1$ to $N$:
			\begin{equation*}
				\frac{\alpha_f \gamma_2(\delta)}{2} \sum_{t=1}^N \mathbb{E}_{\mathcal{D}_2}[I^{(t)}_\delta] \le \mathbb{E}_{\mathcal{D}_1}[\mathcal{L}_L^{(1)}(\bm{\theta}_L^{(1)})] - \mathbb{E}_{\mathcal{D}_1}[\mathcal{L}_L^{(1)}(\bm{\theta}_L^{(N+1)})].
			\end{equation*}
			
			\item \textbf{Bounding the final term.} Let $\mathcal{L}_L^* = \inf_{\bm{\theta}} \mathbb{E}_{\mathcal{D}_2}[\mathcal{L}_L^{(1)}(\bm{\theta})]$. Applying Lemma \ref{lemma:pinsker} again:
			\begin{align}
				\mathbb{E}_{\mathcal{D}_1}[\mathcal{L}_L^{(1)}(\bm{\theta}_L^{(N+1)})] &\ge \mathbb{E}_{\mathcal{D}_2}[\mathcal{L}_L^{(1)}(\bm{\theta}_L^{(N+1)})] -  M\sqrt{\tfrac{1}{2}D_{\mathrm{KL}}} \notag \\ &\ge \mathcal{L}_L^* - M\sqrt{\tfrac{1}{2}D_{\mathrm{KL}}}.
			\end{align} 
			Hence, we have
			\begin{align}
				&\mathbb{E}_{\mathcal{D}_1}[\mathcal{L}_L^{(1)}(\bm{\theta}_L^{(1)})] - \mathbb{E}_{\mathcal{D}_1}[\mathcal{L}_L^{(1)}(\bm{\theta}_L^{(N+1)})] \notag \\ &\le \mathbb{E}_{\mathcal{D}_1}[\mathcal{L}_L^{(1)}(\bm{\theta}_L^{(1)})] - \mathcal{L}_L^* + M\sqrt{\tfrac{1}{2}D_{\mathrm{KL}}}.
			\end{align}
			Combining and dividing by $N$ yields the final bound.
			
		\end{enumerate}
	\end{proof}

	Next, we provide the proof of Corollary \ref{corollary:asymptotic_average}.
	\begin{proof}
		From Theorem \ref{thm:convergence}, we have $\frac{1}{N}\sum_{t=1}^N \Pr_{\hspace{5pt}\mathcal{D}_2}(e^{(t)}\ge\delta) \le \frac{A}{N} + \frac{B}{\sqrt{N}}$, where $A = \frac{2[\mathbb{E}_{\mathcal{D}_1}[\mathcal{L}_L^{(1)}(\bm{\theta}_L^{(1)})] - \mathcal{L}_L^*]}{\alpha_f \gamma_2(\delta)}$ and $B = \frac{2M}{\alpha_f \gamma_2(\delta)}\sqrt{2 D_{\mathrm{KL}}(\mathcal{D}_2\|\mathcal{D}_1)}$ are constants independent of $N$. As $N\to\infty$, both $A/N \to 0$ and $B/\sqrt{N} \to 0$. By the squeeze theorem, the left‑hand side also tends to $0$.
	\end{proof}
	Finally, we provide the proof of Corollary \ref{corollary:pointwise_identical}.
	
	\begin{proof}
		When $\mathcal{D}_1 = \mathcal{D}_2$, we have $D_{\mathrm{KL}}(\mathcal{D}_2\|\mathcal{D}_1)=0$. Substituting into Theorem \ref{thm:convergence} yields $\frac{1}{N}\sum_{t=1}^N \Pr(e^{(t)}\ge\delta) \le \frac{A}{N}$, with $A$ as defined above. Multiplying by $N$, we obtain for every $N\ge 1$, $\sum_{t=1}^N \Pr(e^{(t)}\ge\delta) \le A$. The right‑hand side $A$ is a constant independent of $N$, and all terms are non‑negative. Hence, the partial sums are uniformly bounded and nondecreasing, so they converge to a finite limit as $N\to\infty$. Therefore, $\sum_{t=1}^\infty \Pr(e^{(t)}\ge\delta) < \infty$. A necessary condition for a convergent series is that its terms tend to zero; thus $\Pr(e^{(t)}\ge\delta) \to 0$ as $t\to\infty$.
	\end{proof}
	
\end{document}